\IfFileExists 
    {\jobname.cfg}
    {\input {\jobname.cfg}}
    {\def \CLASS{llncs}}

\documentclass{\CLASS}
\usepackage{gadts}

\title{GADT meet subtyping}

\author{Gabriel Scherer \and Didier R\'emy}

\begin{version}{\Not\RR}
\institute 
   {INRIA, Rocquencourt\thanks
    {Part of this work has been done at \textsc{IRILL}.}}
\end{version}

\csname cmdargs\endcsname
\InputIfFileExists{customize.sty}{}

\providecommand {\Final}{\False}

\Esop
  {\setanswer{hide,later}}
  {\setanswer{inline}}

\Final 
  {\UNXXX}
  {}

\begin{document}

\begin{version}{\RR}
\selectlanguage{francais}
  \makeRR
\selectlanguage{english}
\end{version}
\begin{version}{\Not\RR}

\maketitle 

\begin{abstract}{}
  \RR{\noindent}{}
  While generalized abstract datatypes (\GADT) are now considered
  well-understood, adding them to a language with a notion of
  subtyping comes with a few surprises. What does it mean for a \GADT
  parameter to be covariant?
  The answer turns out to be quite subtle. It involves fine-grained
  properties of the subtyping relation that raise interesting design
  questions.
  We allow variance annotations in \GADT definitions, study their
  soundness, and present a sound and complete algorithm to check them.
  Our work may be applied to real-world ML-like languages with
  explicit subtyping such as OCaml, or to languages with general
  subtyping constraints.
\end{abstract}

\end{version}

\section{Motivation}


In languages that have a notion of subtyping, the interface of
parametrized types usually specifies a \emph{variance}. It defines the
subtyping relation between two instances of a parametrized type from
the subtyping relations that hold between their parameters. For
example, the type $\ty{\alpha}{list}$ of immutable lists is expected
to be \emph{covariant}: we wish $\ty{\sigma}{list} \leq
\ty{\sigma'}{list}$ as soon as $\sigma \leq \sigma'$.

Variance is essential in languages whose programming idioms rely on
subtyping, in particular object-oriented languages. Another reason to
care about variance is its use in the \emph{relaxed value restriction}
\cite{relaxing-the-value-restriction}: while a possibly-effectful
expression, also called an \emph{expansive expression}, cannot be
soundly generalized in ML---unless some sophisticated enhancement of
the type system keeps track of effectful expressions---it is always
sound to generalize type variables that only appear in covariant
positions, which may not classify mutable values.
\begin{version}{\Not\Esop}
  This relaxation uses an intuitive subtyping argument: all
  occurrences of such type variables can be specialized to $\bot$, and
  anytime later, all covariant occurrences of the same variable
  (which are now $\bot$) can be simultaneously replaced by the same
  arbitrary type $\tau$, which is always a supertype of $\bot$.  This
  relaxation of the value-restriction is implemented in OCaml, where
  it is surprisingly useful.
\end{version}
Therefore, it is important for extensions of type definitions, such as
\GADTs, to support it as well through a clear and expressive definition of
parameter covariance.

For example, consider the following \GADT of well-typed expressions:
\label{expr-example}
\begin{lstlisting}[xleftmargin=4em]
type $\ty{\vplus\alpha}{expr}$ =
  | Val : $\alpha \to \ty{\alpha}{expr}$
  | Int : $\tyc{int} \to \ty{\tyc{int}}{expr}$
  | Thunk : $\forall \beta.\,\ty{\beta}{expr} * (\beta \to \alpha) \to \ty{\alpha}{expr}$
  | Prod : $\forall \beta \gamma .\,\ty{\beta}{expr} * \ty{\gamma}{expr} \to \ty{(\beta*\gamma)}{expr}$
\end{lstlisting}
Is it safe to say that $\tyc{expr}$ is covariant in its type
parameter? It turns out that, using the subtyping relation of the
OCaml type system, the answer is ``yes''. But, surprisingly to us, in
a type system with a top type $\top$, the answer would be ``no''.

The aim of this article is to present a sound and complete criterion to
check soundness of parameter variance annotations, for use in a
type-checker. We also discuss the apparent fragility of this criterion
with respect to changes to the subtyping relation (\eg. the presence or
absence of a top type, private types, \etc.), and a different, more robust
way to combine \GADTs and subtyping.

\begin{version}{\Esop}
\paragraph {Note} Due to space restriction, the present article is
only a short version from which many details have been omitted. All
proofs of results presented in this version along with auxiliary
results, as well as further discussions, can be found in the longer
version available online~\cite
{Scherer-Remy:gadts-meet-subtyping@long2012}.
\end{version}

\subsection*{Examples}

Let us first explain why it is reasonable to say that $\ty\alpha{expr}$ is
covariant. Informally, if we are able to coerce
a value of type  $\alpha$ into one of type $\alpha'$ (we write $(\mc{v} :> \alpha')$ to
explicitly cast a value $\mc{v}$ of type $\alpha$ to a value of type
$\alpha'$), then we are also able to transform a value of type $\ty\alpha{expr}$
into one of type $\ty{\alpha'}{expr}$. Here is some pseudo-code\footnote{The
  variables $\beta'$ and $\gamma'$ of the \code{Prod} case are never
  really defined, only justified at the meta-level, making this code
  only an informal sketch.} for the coercion function:
\begin{lstlisting}[xleftmargin=4em]
let coerce : $\ty\alpha{expr} \to \ty{\alpha'}{expr}$ = function
  | Val (v : $\alpha$) -> Val (v :> $\alpha'$)
  | Int n -> Int n
  | Thunk $\beta$ (b : $\beta$ expr) (f : $\beta \to \alpha$) ->
    Thunk $\beta$ b (fun x -> (f x :> $\alpha'$))
  | Prod $\beta$ $\gamma$ ((b, c) : $\ty\beta{expr} * \ty\gamma{expr}$) ->
    (* if $\beta * \gamma \leq \alpha'$, then $\alpha'$ is of the form
       $\beta' * \gamma'$ with $\beta \leq \beta'$ and $\gamma \leq \gamma'$ *)
    Prod $\beta'$ $\gamma'$ ((b :> $\ty{\beta'}{expr}$), (c :> $\ty{\gamma'}{expr}$))
\end{lstlisting}
In the $\mc{Prod}$ case, we make an informal use of something we know
about the OCaml type system: the supertypes of a tuple are all
tuples. By entering the branch, we gain the knowledge that $\alpha$
must be equal to some type of the form $\beta * \gamma$. So from
$\alpha \leq \alpha'$ we know that
$\beta * \gamma \leq \alpha'$. Therefore, $\alpha'$ must itself be
a pair of the form $\beta' * \gamma'$. By covariance of the product,
we deduce that $\beta \leq \beta'$ and $\gamma \leq \gamma'$. This
allows to conclude by casting at types
$\ty{\beta'}{expr}$ and $\ty{\gamma'}{expr}$, recursively.

Similarly, in the $\mc{Int}$ case, we know that $\alpha$ must be an
$\tyc{int}$ and therefore an $\ty{\tyc{int}}{expr}$ is returned. This is
because we know that, in OCaml, no type is above $\tyc{int}$: if $\tyc{int}
\leq \tau$, then $\tau$ must be $\tyc{int}$. 

What we use in both cases is reasoning of the form\footnote {\let \T
  T\let \S \sigma\relax We write $\app \T\bb$ for a type expression
  $\T$ that may contain free occurrences of variables $\bb$ and $\app
  \T {\bar\S}$ for the simultaneous substitution of $\bar\S$ for $\bb$
  in $\T$.}:
``if  
$\app{T}{\bb} \leq \alpha'$, then I know that $\alpha'$ is of the form
$\app{T}{\bb'}$ for some $\bb'$''. We call this an \emph{upward
  closure} property: when we ``go up'' from a $\app{T}\bb$, we only
find types that also have the structure of $T$. Similarly, for
contravariant parameters, we would need a \emph{downward closure}
property: $T$ is downward-closed if $\app{T}\bb \geq \alpha'$ entails
that $\alpha'$ is of the form $\app{T}{\bb'}$.

Before studying a more troubling example, we  define the classic
equality type $\ty{(\alpha, \beta)}{eq}$, and the corresponding
casting function
$\mc{cast}: \forall \alpha \beta. \ty{(\alpha, \beta)}{eq} \to \alpha \to
\beta$: 
\begin{version}{\Not\Esop}
\begin{lstlisting}
  type ($\alpha$, $\beta$) eq =
    | Refl : $\forall \gamma$. ($\gamma$, $\gamma$) eq

  let cast (eqab : $\ty{(\alpha, \beta)}{eq}$) : $\alpha \to \beta$ =
    match eqab with
      | Refl -> (fun x -> x)
\end{lstlisting}
\end{version}
\begin{version}{\Esop}
\begin{lstlisting}[xleftmargin=1em]
type ($\alpha$, $\beta$) eq =                         $\hspace{5.8em}$let cast r =
  | Refl : $\forall \gamma$. ($\gamma$, $\gamma$) eq  $\hspace{2.3em}$  match r with Refl -> (fun x -> x)
\end{lstlisting}
\end{version}
Notice that it would be unsound\footnote{This counterexample is due to
  Jeremy Yallop.} to define $\tyc{eq}$ as covariant, even in only one
parameter. For example, if we had
$\mc{type}~\ty{(\vplus\alpha,\veq\beta)}{eq}$, from any $\sigma \leq
\tau$ we could subtype $\ty{(\sigma,\sigma)}{eq}$ into
$\ty{(\tau,\sigma)}{eq}$, allowing to cast any value of type $\tau$
back into one of type $\sigma$, which is unsound in general.

As a counter-example, the following declaration is incorrect: the type
$\ty{\alpha}{t}$ cannot be declared covariant.
\begin{lstlisting}
  type $\vplus\alpha$ t =
    | K : < m : int > $\to$ < m : int > t
  let v = (K (object method m = 1 end) :> < > t)
\end{lstlisting}
This declaration uses the OCaml object type \code{< m : int >}, which
qualifies objects having a method \code{m} returning an integer. It is
a subtype of object types with fewer methods, in this case the empty
object type \code{< >}, so the alleged covariance of \code{t}, if
accepted by the compiler, would allow us to cast a value of type
\code{< m : int > t} into one of type \code{< > t}. However, from such a 
value, we could wrongly deduce an equality witness
\code{(< >, <m : int>) eq} that allows to cast any
empty object of type \code{< >} into an object of type
\code{< m : int >}, but this is unsound, of course!

\begin{lstlisting}
  let get_eq : $\ty\alpha{t} \to$ ($\alpha$, < m : int >) eq = function
    | K _ -> Refl      $\qquad$ (* locally $\alpha =\,$< m : int > *)
  let wrong : < > -> < m : int > =
    let eq : (< >, < m : int >) eq = get_eq v in
    cast eq
\end{lstlisting}
It is possible to reproduce this example using a different feature of
the OCaml type system named \emph{private type
  abbreviation}\footnote{This counterexample is due to Jacques
  Garrigue.}: a module using a type $\mc{type}~\tyc{t}~\mc{=}~\tau$
\emph{internally} may describe its interface as
$\mc{type}~\tyc{t}~\mc{=}~\mc{private}~\tau$. This is a compromise
between a type abbreviation and an abstract type: it is possible to
cast a value of type $\tyc{t}$ into one of type $\tau$, but not,
conversely, to construct a value of type $\tyc t$ from one of type
$\tau$. In other words, $\tyc{t}$ is a strict subtype of $\tau$: we
have $\tyc{t} \leq \tau$ but not $ \tyc{t} \geq\tau$. Take for example
\code{type file_descr = private int}: this semi-abstraction is useful
to enforce invariants by restricting the construction of values of
type \code{file_descr}, while allowing users to conveniently and
efficiently destruct them for inspection at type \code{int}.
\begin{version}{\Esop}
  Using an unsound but quite innocent-looking covariant \GADT
  datatype, one is able to construct a function to cast any integer
  into a \code{file_descr}, which defeats the purpose of this
  abstraction---see the extended version of this article for the full
  example.
\end{version}

\begin{version}{\Not\Esop}
Unsound \GADT covariance declarations would defeat the purpose of such
private types: as soon as the user gets one element of the private
type, she could forge values of this type, as illustrated by the code
below.
\begin{lstlisting}
  module M = struct
    type file_descr = int
    let stdin = 0
    let open = ...
  end : sig
    type file_descr = private int
    val stdin : file_descr
    val open : string -> (file_descr, error) sum
  end

  type $\vplus\alpha$ t =
    | K : priv -> M.file_descr t

  let get_eq : $\ty\alpha{t}$ -> ($\alpha$, M.file_descr) eq = function
    | K _ -> Refl

  let forge : int -> M.file_descr =
    fun (x : $\tyc{int}$) -> cast (get_eq p) M.stdin
\end{lstlisting}
\end{version}

The difference between the former, correct \code{Prod} case and those
two latter situations with unsound variance is the notion of upward
closure. The types $\alpha * \beta$ and $\tyc{int}$ used in the
correct example were upward-closed. On the contrary, the private type
\Esop{\code{file_descr}}{\code{M.file_descr}} has a distinct supertype
\code{int}, and similarly the object type \code { < m:int > } has
a supertype \code {< >} with a different structure (no method
\code m).

In this article, we formally show that these notions of upward and
downward-closure are the key to a sound variance check for \GADTs. We
start from the formal development of Simonet and Pottier
\cite{simonet-pottier-hmg-toplas}, which provides a general soundness
proof for a language with subtyping and a very general notion of \GADT
expressing arbitrary constraints---rather than only type
equalities. By specializing their correctness criterion, we can
express it in terms of syntactic checks for closure and variance, that
are simple to implement in a type-checker.

\subsection*{The problem of non-monotonicity}

There is a problem with those upward or downward closure assumptions: while
they hold in core ML, with strong inversion theorems,
they are non-monotonic properties: they are not necessarily preserved by
extensions of the subtyping lattice. For example, OCaml has a concept of
\emph{private types}: a type specified by
$\mc{type}~\tyc{t}~=~\mc{private}~\tau$ is a new semi-abstract type smaller
than $\tau$ ($\tyc{t} \leq \tau$ but $\tyc{t} \ngeq \tau$), that can be
defined a posteriori for any type. Hence, no type is downward-closed
\emph{forever}. That is, for any type $\tau$, a new, strictly
smaller type may always be defined in the future.

This means that closure properties of the OCaml type system are
relatively weak: no type is downward-closed\footnote{Except types that
  are only defined privately in a module and not exported: they exist
  in a ``closed world'' and we can check, for example, that they are
  never used in a \code{private} type definition.} (so instantiated
\GADT parameters cannot be contravariant), and arrow types are not
upward-closed as their domain should be downward-closed. Only purely
positive algebraic datatypes are upward-closed. The subset of \GADT
declarations that can be declared covariant today is small, yet, we
think, large enough to capture a lot of useful examples, such as
$\ty{\alpha}{expr}$ above.

\subsection*{Giving back the freedom of subtyping}

It is disturbing that the type system should rely on non-monotonic
properties: if we adopt the correctness criterion above, we must be
careful in the future not to enrich the subtyping relation too
much.

\XXX{Mention that this ``non-monotonous'' behavior is due to
  a new ``inversion'' in the type system, for which we didn't have
  good example before (new thing)?}\XXX[GS]{I'm not sure this is the
  right place to speak about that. I think we should rather do that in
  the conclusion section.}

Consider \code{private} types for example: one could imagine
a symmetric concept of a type that would be strictly \emph{above}
a given type $\tau$; we will name those types \code{invisible} types
(they can be constructed, but not observed). Invisible types and \GADT
covariance seem to be working against each other: if the designer adds
one, adding the other later  will be difficult.

A solution to this tension is to allow the user to \emph{locally}
guarantee negative properties about subtyping (what is \emph{not}
a subtype), at the cost of selectively abandoning the corresponding
flexibility.  Just as object-oriented languages have \code{final}
classes that cannot be extended any more, we would like to be able to define
some types as \code{public} (respectively \code{visible}), that cannot
later be made \code{private} (resp. \code{invisible}). Such
declarations would be rejected if the defining type already has
subtypes (\eg. an object type), and would forbid further declarations
of types below (resp. above) the defined type, effectively guaranteeing
downward (resp. upward) closure. Finally, upward or downward closure
is a semantic aspect of a type that we must have the freedom to
publish through an interface: abstract types could optionally be
declared \code{public} or \code{visible}.

\subsection*{Another approach: subtyping constraints}
\label{sec/introduction-subtyping-contraints}

Getting fine variance properties out of \GADT is
difficult because they correspond to type equalities which, to
a first approximation, use their two operands both positively and
negatively. One way to get an easy variance check is to encourage
users to \emph{change} their definitions into different ones that are
easier to check.  For example, consider the following redefinition of
$\ty{\alpha}{expr}$ (in a speculative extension of OCaml with subtyping
constraints):
\begin{lstlisting}[basewidth=0.40em,columns=fixed]
type +$\ty{\alpha}{expr}$ =
| Val : $\forall \alpha. \alpha \to \ty{\alpha}{expr}$
| Int : $\forall \alpha [\alpha\mathord{\geq}\tyc{int}]. \tyc{int} \to \ty{\alpha}{expr}$
| Thunk : $\forall \beta.\,\ty{\beta}{expr} * (\beta \to \alpha) \to \ty{\alpha}{expr}$
| Prod : $\forall \alpha \beta \gamma [\alpha \mathord{\geq} \beta*\gamma].\,(\ty{\beta}{expr} * \ty{\gamma}{expr}) \to \ty{\alpha}{expr}$
\end{lstlisting}
It is now quite easy to check that this definition is covariant, since
all type equalities $\alpha = \app{T_i}{\bb}$ have been replaced by
inequalities $\alpha \geq \app{T_i}{\bb}$ which are preserved when
replacing $\alpha$ by a subtype $\alpha' \geq \alpha$---we explain
this more formally in
\S\ref{sec/gadts-with-subtyping-constraints}. This variation on
\GADTs, using subtyping instead of equality constraints, has been
studied by Emir \etal.~\cite{csharp-generalized-constraints} in the
context of the \csharp programming language.

But isn't such a type definition less useful than the previous one,
which had a stronger constraint? We will discuss this choice in more
detail in~\S\ref{sec/gadts-with-subtyping-constraints}.

\begin{version}{\Not\Esop}
\subsection*{On the importance of variance annotations}

Being able to specify the variance of a parametrized datatype is
important at abstraction boundaries: one may wish to define a program
component relying on an \emph{abstract} type, but still make certain
subtyping assumptions on this type. Variance assignments provide
a framework to specify such a semantic interface with respect to
subtyping. When this abstract type dependency is provided by an
encapsulated implementation, the system must check that the provided
implementation indeed matches the claimed variance properties.

Assume the user specifies an abstract type
\begin{lstlisting}[xleftmargin=4em]
module type S = sig
  type $(+\alpha)$ collection
  val empty : unit -> $\alpha$ collection
  val app : $\alpha$ collection -> $\alpha$ collection -> $\alpha$ collection
end
\end{lstlisting}
and then implements it with linked lists
\begin{lstlisting}[xleftmargin=4em]
module C : S = struct
  type $+\alpha$ collection =
    | Nil of $\tyc{unit}$
    | Cons of $\alpha\,*\,\ty{\alpha}{collection}$
  let empty () = Nil ()
end
\end{lstlisting}
The type-checker will accept this implementation, as it has the
specified variance. On the contrary,
\begin{lstlisting}[xleftmargin=4em]
type $+\alpha$ collection = $\ty{(\ty{\alpha}{list})}{ref}$  
let empty () = ref []
\end{lstlisting}
would be rejected, as $\tyc{ref}$ is invariant.
In the following definition:
\begin{lstlisting}[xleftmargin=4em]
let nil = C.empty ()
\end{lstlisting}
the right hand-side is not a value, and is therefore not generalized
in presence of the value restriction; we get a monomorphic type,
$\ty{?\alpha}{t}$, where $?\alpha$ is a yet-undetermined type
variable. The relaxed value restriction
\cite{relaxing-the-value-restriction} indicates that it is sound to
generalize $?\alpha$, as it only appears in covariant
positions. Informally, one may unify $?\alpha$ with $\bot$, add an
innocuous quantification over $\alpha$, and then generalize
$\forall \alpha. \ty{\bot}{t}$ into $\forall \alpha. \ty{\alpha}{t}$
by covariance---assuming a lifting of subtyping to polymorphic type
schemes.

The definition of \code{nil} will therefore get generalized in
presence of the relaxed value restriction, which would not be the case
if the interface \texttt S had specified an invariant type.
\end{version}

\subsection*{Related work}


\begin{version}{\Esop}
  Simonet and Pottier \cite{simonet-pottier-hmg-toplas} have studied
  \GADTs in a general framework HMG(X), inspired by HM(X). They were
  interested in typing inference using constraints, so considered
  \GADTs with arbitrary constraints rather than type equalities, and
  considered the case of subtyping with applications to information
  flow security in mind. We instantiate their general framework, which
  allows us to reuse their dynamic semantics and syntactic proofs of
  soundness, and concentrate only on the static semantics, proving
  that we meet the requirements they impose on the parametrized
  models.

  Their soundness criterion is formulated in very general terms as
  a constraint entailment problem. In contrast, our specialized study
  of the case of equality and subtyping led to a refined, more
  syntactic, criterion. This provides a more practically implementable
  check for type definitions, and reveals the design issues
  surrounding $v$-closed constructors that were not apparent in their
  work.
\end{version}
\begin{version}{\Not\Esop}
When we encountered the question of checking variance annotations on
\GADTs, we expected to find it already discussed in the
literature. The work of Simonet and Pottier
\cite{simonet-pottier-hmg-toplas} is the closest we could find. It was
done in the context of finding good specification for \emph{type
  inference} of code using \GADTs, and in this context it is natural
to embed some form of constraint solving in the type inference
problem. From there, Simonet and Pottier generalized to a rich notion
of \GADTs defined over arbitrary constraints, in presence of
a subtyping relation, justified in their setting by potential
applications to information flow checking.

They do not describe a particular type system, but a parametrized
framework HMG(X), in the spirit of the inference framework HM(X). In
this setting, they prove a general soundness result, applicable to all
type systems which satisfy their model requirements. We directly reuse
this soundness result, by checking that we respect these requirements
and proving that their condition for soundness is met. This allows us
to concentrate purely on the static semantics, without having to
define our own dynamic semantics to formulate subject reduction and
progress results.

Their soundness requirement is formulated in terms of a general
constraint entailment problem involving arbitrary
constraints. Specializing this to our setting is simple, but expressing
it in a form that is amenable to mechanical verification is surprisingly
harder---this is the main result of this paper. Furthermore, at their
level of generality, the design issues related to subtyping of \GADTs,
in particular the notion of upward and downward-closed type
constructors, were not apparent. Our article is therefore not only
a specialized, more practical instance of their framework, but also
raises new design issues.
\end{version}


\medskip 

\begin{version}{\Esop}
  Emir, Kennedy, Russo and Yu~\cite{csharp-generalized-constraints}
  studied the soundness of an object-oriented calculus with subtyping
  constraints on classes and methods. Previous work
  \cite{gadt-and-oop} had established the correspondence between
  equality constraints on methods in an object-oriented style and GADT
  constraints on type constructors in functional style. Through this
  correspondence, their system matches our presentation of \GADTs with
  subtyping constraints and easier variance assignment, mentioned in
  the introduction~(\S\ref{sec/introduction-subtyping-contraints}) and
  detailed in~\S\ref{sec/gadts-with-subtyping-constraints}. They do
  not encounter the more delicate notion of upward and downward
  closure.
  
  Those two approaches (subtyping constraints with easy
  variance check, stronger equality constraints with more delicate
  variance check) are complementary and have different convenience
  trade-offs. In their system with explicit-constraint definitions
  and implicit subtyping, the subtyping-constraint solution is the
  most convenient, while our ML setting provides incentives to study
  the other solution.
\end{version}

\begin{version}{\Not\Esop}
  The other major related work, by Emir, Kennedy, Russo and
  Yu~\cite{csharp-generalized-constraints}, studies the soundness
  of having subtyping constraints on classes and methods of an
  object-oriented type system with generics (parametric
  polymorphism). Previous work~\cite{gadt-and-oop} had already
  established the relation between the \GADT style of having type
  equality constraints on data constructors and the desirable
  object-oriented feature of having type equality constraints on
  object methods. This work extends it to general subtyping
  constraints and develops a syntactic soundness proof in the context
  of a core type system for an object-oriented languages with
  generics.

The general duality between the ``sums of data'' prominent in
functional programming and ``record of operations'' omnipresent in
object-oriented programming is now well-understood. Yet, it is
surprisingly difficult to reason on the correspondence between \GADTs
and generalized method constraints; an application that is usually
considered to require \GADTs in a functional style (for example
a strongly-typed eval $\ty{\alpha}{expr}$ datatype and its associated
\code{eval} function) is simply expressed in idiomatic object-oriented
style without specific constraints\footnote{There is a relation
  between this way of writing a strongly typed eval function and the
  ``finally tagless'' approach \cite{finally-tagless} that is known to require
  only simple ML types.}, while the simple
$\mc{flatten} : \forall \alpha,\, \ty{\ty{\alpha}{list}}{list} \to \ty{\alpha}{list}$
requires an equality or subtyping constraint when expressed in
object-oriented style.

These important differences of style and applications make it difficult
to compare our present work with this one. Our understanding of this
system is that a subtyping constraint of the form $X \leq Y$ is
considered to be a negative occurrence of $X$, and a positive occurrence
of $Y$; this means that equality constraints (which are conjunctions
of a $\rel\leq$ constraint and a $\rel\geq$ constraints) always impose invariance
on their arguments. Checking correctness of constraints with this
notion of variance is simpler than with our upward and
downward-closure criterion, but also not as expressive. It
corresponds, in our work, to the idea of \GADTs with subtyping
constraint mentioned in the introduction and that we detail in
\S\ref{sec/gadts-with-subtyping-constraints}.

The design trade-off in this related work is different from our
setting; the reason why we do not focus on this solution is that it
requires explicit annotations at the GADT definition site, and more
user annotations in pattern matching in our system where subtyping is
explicitly annotated, while convertibility is implicitly inferred by
unification. On the contrary, in an OOP system with explicit
constraints and implicit subtyping, this solution has the advantage of
user convenience.

We can therefore see our present work as a different choice in the
design space: we want to allow a richer notion of variance assignment
for type \emph{equalities}, at the cost a higher complexity for those
checks. Note that the two directions are complementary and are both
expressed in our formal framework.
\end{version}

\section{A formal setting}\label{formal_setting}

We define a core language for Algebraic Datatypes (ADT) and, later,
Generalized Algebraic Datatypes (\GADTs), that is an instance of the
parametrized HMG(X) system of Simonet and Pottier
\cite{simonet-pottier-hmg-toplas}. We refine their framework by using
variances to define subtyping, but rely on their formal description for
most of the system, in particular the static and dynamic semantics. We
ultimately rely on their type soundness proof, by rigorously showing
(in the next section) that their requirements on datatype definitions
for this proof to hold are met in our extension with variances.

\subsection{Atomic subtyping}

Our type system defines a subtyping relation between ground types,
parametrized by a reflexive transitive relation between base constant
types (\code{int}, \code{bool}, etc.). Ground types consist of a set
of base types $\tyc{b}$, function types $\tau_1 \to \tau_2$, product
types $\tau_1 * \tau_2$, and a set of algebraic datatypes
$\ty{\bs}{t}$. (We write $\bar \sigma$ for a sequence of types
$(\sigma_i)_\iI$.)  We use prefix notation for datatype parameters, as
is the usage in ML. Datatypes may be user-defined by toplevel
declarations of the form:
\begin{lstlisting}
  type $\ty{\Gva}{t}$ =
    | K$_1$ of $\app{\tau^1}{\ba}$
    | ...
    | K$_n$ of $\app{\tau^n}{\ba}$
\end{lstlisting}
This is a disjoint sum: the constructors \code{K$_c$} represent all 
possible cases and each type
$\app {\tau^c} \ba$ is the domain of the constructor
\code{K$_c$}. Applying it to an argument $e$ of a corresponding ground
type $\app{\tau}{\bs}$ constructs a term of type $\ty{\bs}{t}$. Values
of this type are deconstructed using pattern matching clauses of the
form $\mc{K}_c~x \to e$, one for each constructor.

The sequence $\Gva$ is a binding list of type variables $\alpha_i$
along with their \emph{variance annotation} $v_i$, which is a marker
among the set $\set{\vplus, \vminus, \veq, \virr}$. We may associate
a relation a relation $\rel{\prec_v}$ between types to each variance
$v$:
\begin{itemize}
\item $\prec_{\vplus}$ is the \emph{covariant} relation $\rel\leq$;
\item $\prec_{\vminus}$ is the \emph{contravariant} relation $\rel\geq$,
the symmetric of $\rel\leq$;
\item $\prec_{\veq}$ is the \emph{invariant} relation $\rel{=}$, defined as
the intersection of $\rel\leq$ and $\rel\geq$;
\item $\prec_{\virr}$, is the \emph{irrelevant} relation $\rel\Join$,
  the full relation such that $\sigma \Join \tau$ holds for all types
  $\sigma$ and $\tau$.
\end{itemize}

Given a reflexive transitive relation $\rel\leqslant$ on base
types, the subtyping relation on ground types $\rel\leq$ is defined by the
inference rules of Figure~\ref {fig/subtyping}, which, in particular, give
their meaning to the variance annotations $\Gva$. The judgment
$\mc{type}~\ty{\Gva}{t}$ simply means that the type constructor
$\tyc{t}$ has been previously defined with the variance annotation
$\Gva$. 
\begin{mathparfig}{fig/subtyping}{Subtyping relation}
\inferrule{ }
  {\sigma \leq \sigma}

\inferrule
  {\sigma_1 \leq \sigma_2 \\ \sigma_2 \leq \sigma_3}
  {\sigma_1 \leq \sigma_3}

\inferrule
    {\tyc{b} \leqslant \tyc{c}}
    {\tyc{b} \leq \tyc{c}}

\inferrule
  {\sigma \geq \sigma' \\ \tau \leq \tau'}
  {\sigma \to \tau \leq \sigma' \to \tau'}
  \quad \and

\inferrule
  {\sigma \leq \sigma'\\ \tau \leq \tau'}
  {\sigma * \tau \leq \sigma' * \tau'}

\inferrule
  {\mc{type}~\ty{\Gva}{t} \\
   \forall i,\,\sigma_i \prec_{v_i} \sigma'_i}
  {\ty{\bs}{t} \leq \ty{\bs'}{t}}
\end{mathparfig}
\XXX{Rules for subtyping should have names}
Notice that the rules for arrow and product types can be subsumed by
the rule for datatypes, if one consider them as special datatypes
(with a specific dynamic semantics) of variance $(\vminus, \vplus)$
and $(\vplus, \vplus)$, respectively. For this reason, the following
definitions will not explicitly detail the cases for arrows and
products.

\begin{version}{\Not\Esop}
Finally, it is routine to show that the rules for reflexivity and
transitivity are admissible, by pushing them up in the derivation
until the base cases $\tyc{b} \leqslant \tyc{c}$, where they can be
removed as $\rel\leqslant$ is assumed to be reflexive and transitive. 
Removing reflexivity and transitivity provides us with an equivalent
syntax-directed judgment having powerful inversion principles: if
$\ty \bs t \leq \ty {\bs'} t$ and
$\mc{type}~\ty{\Gva}{t}$, then one can deduce that for each
$i$, $\sigma_i \prec_{v_i} \sigma'_i$.\label{subtyping_inversion} 
\end{version}

\begin{version}{\Esop}
  As usual in subtyping systems, we could reformulate our judgment in
  a syntax-directed way, to prove that it admits good inversion
  properties: if $\ty \bs t \leq \ty {\bs'} t$ and
  $\mc{type}~\ty{\Gva}{t}$, then one can deduce that for each $i$,
  $\sigma_i \prec_{v_i} \sigma'_i$.\label{subtyping_inversion}
\end{version}

\begin{version}{\Not\Esop}
\XXX[GS]{Didier, please review this paragraph on the difference between
  equivalence and syntactic equality} We insist that our equality
relation $\rel\veq$ is here a derived concept, defined from the
subtyping relation $\rel\leq$ as the ``equiconvertibility'' relation
$\rel{\leq \cap \geq}$; in particular, it is not defined as the usual
syntactic equality. If we have both $b_1 \leqslant b_2$ and
$b_1 \leqslant b_2$ in our relation on base types, for two distinct
base types $b_1$ and $b_2$, we have $b_1 = b_2$ as types, even though they are
syntactically distinct. This choice is inspired by the previous work
of Simonet and Pottier.
\end{version}

\begin{version}{\Not\Esop}
\paragraph{On the restriction of atomic subtyping}\label{atomic_subtyping_restriction}

The subtyping system demonstrated above is called ``atomic''. If two
head constructors are in the subtyping relation, they are either
identical or constant (no parameters). Structure-changing subtyping
occurs only at the leaves of the subtyping derivations.

While this simplifies the meta-theoretic study of the subtyping
relation, this is too simplifying for real-world type systems that use
non-atomic subtyping relations. In our examples using the OCaml type
system, \code{private} type were a source of non-atomic subtyping: if
you define
$\mc{type}~\ty{\alpha}{t_2}~=~\mc{private}~\ty{\alpha}{t_1}$, the head
constructors $t_1$ and $t_2$ are distinct yet in a subtyping
relation. If we want to apply our formal results to the design of such
languages, we must be careful to isolate any assumption on this atomic
nature of our core formal calculus.

The aspect of non-atomic subtype relations we are interested in is the
notion of $v$-closed constructor.
We have used this notion informally in
the first section (in OCaml, product types are $\vplus$-closed);
we now defined it formally. 
\begin{definition}[Constructor closure]
\label{def/v-closed}
A type constructor $\ty{\ba}{t}$ is $v$-closed if, for any type sequence
$\bs$ and type $\tau$ such that $\ty{\bs}{t} \prec_v \tau$ hold, then
$\tau$ is necessarily equal to $\ty{\bs'}{t}$ for some $\bs'$.
\end{definition}
In our core calculus, all
type constructors are $v$-closed for any $v \neq \virr$, but we will
still mark this hypothesis explicitly when it appears in typing
judgments; this let the formal results be adapted more easily to
a non-atomic type system.

It would have been even more convincing to start from a non-atomic
subtyping relation. However, the formal system of Simonet and Pottier,
whose soundness proof we ultimately reuse, restricts subtyping
relations between (G)ADT type to atomic subtyping. We are confident
their proof (and then our formal setting) can be extended to cover the
non-atomic case, but we have left this extension to future work.
\end{version}

\begin{version}{\Esop}
\paragraph{On the restriction of atomic subtyping}\label{atomic_subtyping_restriction}

Our typing relation reproduces a simplification that is present in the
formulation of Simonet and Pottier: it is \emph{atomic} in the sense
that two non-constant type constructors in the subtyping relation are
always identical. We are confident their proof
(and then our formal setting) can be extended to cover the non-atomic
case, but we have left this extension to future work.

Richer type systems, for example if they have bottom or top types or,
in the case of the OCaml type system, \code{private} types and object
types, have a non-atomic subtyping relations. To be able to extend our
work to such settings, we have carefully marked each use of atomic
subtyping in the formal development with an hypothesis of $v$-closure,
defined below. In the case of atomic subtyping, all types are
$v$-closed.

\begin{definition}[Constructor closure]
\label{def/v-closed}
A type constructor $\ty{\ba}{t}$ is $v$-closed if, for any type sequence
$\bs$ and type $\tau$ such that $\ty{\bs}{t} \prec_v \tau$ hold, then
$\tau$ is necessarily equal to $\ty{\bs'}{t}$ for some $\bs'$.
\end{definition}
\end{version}

\subsection{The algebra of variances}

If we know that $\ty{\bs}{t} \leq \ty{\bs'}{t}$, that is $\ty{\bs}{t}
\prec_{\vplus} \ty{\bs'}{t}$, and the constructor $\tyc{t}$ has variable
$\bar{v\alpha}$, an inversion principle tells us that for each $i$,
$\sigma_i \prec_{v_i} \sigma'_i$. But what if we only know $\ty \bs
t \prec_{u} \ty {\bs'} t$ for some variance $u$ different from
$\rel\vplus$? If $u$ is $\rel\vminus$, we get the reverse relation
$\sigma_i \succ_{v_i} \sigma'_i$. If $u$ is $\rel\virr$, we get
$\sigma_i \Join \sigma'_i$, that is, nothing. This outlines
a \emph{composition} operation on variances $\varcomp u {v_i}$, such
that if $\ty \bs t \prec_{u} \ty {\bs'} t$ then $\sigma_i
\prec_{\varcomp u {v_i}} \sigma'_i$ holds. It is defined by the
following table:
$$
\def \C#1{\multicolumn{1}{C}{#1}}
\begin{tabular}{R|c@{}|C|C|C|C|l}
  \varcomp v w && \C{\veq} & \C{\vplus} & \C{\vminus} & \C{\virr} & w
\\ 
\cline{1-1}\cline{3-7} \noalign{\vskip \doublerulesep}\cline{1-1}\cline{3-6}
  \veq    && \veq  & \veq     & \veq    & \virr \\ \cline{3-6}
  \vplus  && \veq  & \vplus   & \vminus & \virr \\ \cline{3-6}
  \vminus && \veq  & \vminus  & \vplus  & \virr \\ \cline{3-6}
  \virr   && \virr & \virr    & \virr   & \virr \\ \cline{3-6}
  v       &\multicolumn{5}{C}{}
  \end{tabular}
$$
This operation is associative and commutative. Such an operator, and
the algebraic properties of variances explained below, have already
been used by other authors, for example
\cite{polarized-subtyping-for-sized-types}.

There is a natural order relation between \emph{variances}, which
is the \emph{coarser-than} order between the corresponding relations:
$v \leq w$ if and only if 
$\rel{\prec_v} \supseteq \rel{\prec_w}$; \ie. if and only if, for all
$\sigma$ and $\tau$, $\sigma \prec_w \tau$ implies
$\sigma \prec_v \tau$.\footnote{The reason for this order reversal is
  that the relations occur as hypotheses, in negative position, in
  definition of subtyping: if we have $v \leq w$ and
  $\mc{type}~\ty{v\alpha}{t}$, it is safe to assume
  $\mc{type}~\ty{w\alpha}{t}$: $\sigma \prec_w \sigma'$ implies
  $\sigma \prec_v \sigma'$, which implies
  $\ty{\sigma}{t} \leq \ty{\sigma'}{t}$. One may also see it, as Abel
  notes, as an ``information order'': knowing that
  $\sigma \prec_{\vplus} \tau$ ``gives you more information'' than
  knowing that $\sigma \prec_{\virr} \tau$, therefore
  $\virr \leq \vplus$.}
This reflexive, partial order is described by the following lattice diagram:
$$
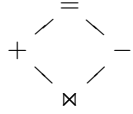
\xymatrix@=1.4ex{
 & \veq \ar@{-}[dr] & \\
\vplus \ar@{-}[ur] \ar@{-}[dr] & & - \\
& \virr \ar@{-}[ur] & \\
}
$$
That is,  all variances are smaller than $\veq$ and bigger than $\virr$. 

\begin{version}{\False}
Conversely, if we assume that $D[C[\_]]$ has variance $u$, and know
that $D[\_]$ has variance $w$, what must be the variance of $C[\_]$
alone ? For a fixed $w$, the composition operation
$v \mapsto \varcomp v w$ has a ``quasi-inverse'', which we will write
$u \mapsto \vardiv u w$. We say ``quasi'' because it does not verify
$\vardiv {(\varcomp v w)} w = v$, but only the weaker
$\vardiv {(\varcomp v w)} w \leq v$. More generally, Abel gives the
following connection:
  \[ \forall u, v, w,\quad u \leq \varcomp v w \iff \vardiv u w \leq v \]
The operation table for $\vardiv u w$ is the following:
$$
  \begin{tabular}{r|c|c|c|c|l}
    \multicolumn{1}{r}{$\vardiv u w$}
    & \multicolumn{1}{c}{$\veq$}
    & \multicolumn{1}{c}{$\vplus$}
    & \multicolumn{1}{c}{$\vminus$}
    & \multicolumn{1}{c}{$\virr$}
    & $w$
    \\ \cline{2-5}
    $\veq$    & $\veq$  & $\veq$    & $\veq$    & $\veq$ \\ \cline{2-5}
    $\vplus$  & $\virr$ & $\vplus$  & $\vminus$ & $\veq$ \\ \cline{2-5}
    $\vminus$ & $\virr$ & $\vminus$ & $\vplus$  & $\veq$ \\ \cline{2-5}
    $\virr$   & $\virr$ & $\virr$   & $\virr$   & $\veq$ \\ \cline{2-5}
    \multicolumn{1}{r}{$u$}
    & \multicolumn{1}{c}{}
    & \multicolumn{1}{c}{}
    & \multicolumn{1}{c}{}
    & \multicolumn{1}{c}{}
    & \\
  \end{tabular}
$$
This operation is associative, and further we have that
$\vardiv{(\vardiv u w)} w' = \vardiv u {(\varcomp w w')}$. It is also
monotonous, but contravariant in its right argument: if $v \leq v'$
and $w \geq w'$, then $\vardiv v w \leq \vardiv {v'} {w'}$.
\end{version}

From the order lattice on variances we can define join $\vee$ and meet
$\wedge$ of variances: $v \vee w$ is the biggest variance such that
$v \vee w \leq v$ and $v \vee w \leq v$; conversely, $v \wedge w$ is
the lowest variance such that $v \leq v \wedge w$ and
$w \leq v \wedge w$. Finally, the composition operation is monotonous:
if $v \leq v'$ then $\varcomp w v \leq \varcomp w {v'}$
(and $\varcomp v w \leq \varcomp {v'} w$).

We will frequently manipulate vectors $\bar{v\alpha}$, of variable
associated with variances, which correpond to the ``context'' $\Gamma$
of a type declaration. We extend our operation pairwise on those
contexts: $\Gamma \vee \Gamma'$ and $\Gamma \wedge \Gamma'$, and the
ordering between contexts $\Gamma \leq \Gamma'$. We also extend the
variance-dependent subtyping relation $\rel{\prec_v}$, which becomes
an order $\rel{\prec_\Gamma}$ between vectors of type of the same
length: $\bs \prec_{\bar{v\alpha}} \bs'$ holds when for all $i$ we
have $\sigma_i \prec_{v_i} \sigma'_i$.

\subsection{Variance assignment in ADTs}

\paragraph{A counter-example}

To have a sound type system, some datatype declarations must be
rejected. 
Assume (only for this example) that we have two base types $\tyc{int}$ and
$\tyc{bool}$ such that $\tyc{bool} \leqslant \tyc{int}$ and 
$\tyc{int} \not\leqslant \tyc{bool}$. 
Consider the following type declaration: 
\begin{lstlisting}
  type $\ty{(\vplus\alpha, \vplus\beta)}{t}$ =
    | Fun of $\alpha \to \beta$
\end{lstlisting}
If it were accepted, we could build type the following program that deduces
from the $(+\alpha)$ variance that $\ty{(\tyc{bool},\tyc{bool})}{t} \leq
\ty{(\tyc{int},\tyc{bool})}{t}$; that is, we could turn the identity
function of type $\tyc{bool}\to \tyc{bool}$ into one of type $\tyc{int}\to
\tyc{bool}$ and then turns an integer into a boolean:
\begin{lstlisting}
  let three_as_bool : $\tyc{bool}$ = 
    match (Fun (fun x -> x) : $\ty{(\tyc{bool}, \tyc{bool})}{t}$ :> $\ty{(\tyc{int}, \tyc{bool})}{t}$) with
      | Fun (danger : $\tyc{int} \to \tyc{bool}$) -> danger 3
\end{lstlisting}

\paragraph{A requirement for type soundness}

We say that the type $\mc{type}~\ty{\Gva}{t}$ defined by the
constructors $(\mc{K}_c~\mc{of}~\app{\tau^c}{\ba})_{\cC}$
is \emph{well-signed} if
\[ \forall \cC, \forall \bs, \forall \bs', \quad
   \ty{\bs}{t} \leq \ty{\bs'}{t} \implies 
     \app{\tau^c}{\bs} \leq \app{\tau^c}{\bs'}
\]
The definition of $\ty{(+\alpha,+\beta)}{t}$ is not well-signed
because we have $\ty{(\bot,\bot)}{t} \leq \ty{(\tyc{int},\bot)}{t}$ 
according to the variance declaration, but we do not have the corresponding
conclusion $(\tyc{int} \to \bot) \leq (\bot \to \bot)$.

This is a simplified version, specialized to simple algebraic
datatypes, of the soundness criterion of Simonet and Pottier. They
proved that this condition is \emph{sufficient}\footnote{It turns out
  that this condition is not \emph{necessary} and can be slightly
  weakened: \Esop {we discuss this in the extended version of this
    article} {we will discuss that later
    (\ref{simonet-pottier-not-complete})}.} for soundness: if all
datatype definitions accepted by the type-checker are well-signed,
then both subject reduction and progress hold---for their static and
dynamic semantics, using the subtyping relation $(\leq)$ we have
defined.
\paragraph {A judgment for variance assignment}

When reformulating the well-signedness requirement of Simonet and
Pottier for simple ADT, in our specific case where the subtyping
relation is defined by variance, it becomes a simple check on the
variance of type definitions. Our example above is unsound as its
claims $\alpha$ covariant while it in fact appears in negative
position in the definition.

\begin{version}{\Not\Esop}
In the context of higher-order subtyping
\cite{polarized-subtyping-for-sized-types}, where type abstractions
are first-class and annotated with a variance
($\lambda v\alpha. \tau$), it is natural to present this check as
a kind checking of the form $\G \der \tau : \kappa$, where
$\G$ is a context $\Gva$ of type variables \emph{associated with
  variances}. For example, if $\vplus\alpha \der \tau : \star$ is provable, it
is sound to consider $\alpha$ covariant in $\tau$. In the context of
a simple first-order monomorphic type calculus, this amounts to
a \emph{monotonicity check} on the type $\tau$ as defined by
\cite{csharp-generalized-constraints}.
Both approaches use judgments of a peculiar form where the
\emph{context} changes when going under a type constructor: to check
$\G \der \sigma \to \tau$, one checks $\G \der \tau$ but
$(\vardiv \G \vminus) \der \sigma$, where $\vardiv \G \vminus$
reverses all the variances in the context $\G$
(turns $\rel\vminus$ into $\rel\vplus$ and conversely). Abel gives an
elegant presentation of this inversion $/$ as an algebraic operation
on variances, a quasi-inverse such that $\vardiv u v \leq w$ if and
only if $u \leq \varcomp w v$. This context change is also reminiscent
of the \emph{context resurrection} operation of the literature on
proof irrelevance (in the style of \cite{pfenning:intextirr}
for example).
\end{version}

\begin{mathparfig}{fig/variance}{Variance assignment}
\inferrule[vc-Var]
  {w\alpha \in \G\\ w \geq v}
  {\G \der \alpha : v}

\inferrule[vc-Constr]
  {\G \der \mc{type}~\ty{\bar{w\alpha}}{t}\\
   \forall i,\ \G \der \sigma_i : \varcomp v {w_i}}
  {\G \der \ty{\bar{\sigma}}{t} : v}
\end{mathparfig}

\begin{version}{\Not\Esop}
  We chose an equivalent but more conventional style where the context
  of subderivation does not change: instead of a judgment
  $\G \der \tau$ that becomes ${(\vardiv \G u) \der \sigma}$ when
  moving to a subterm of variance $u$, we define a judgment of the
  form $\G \der \tau : v$, that evolves into
  $\G \der \sigma : (\varcomp v u)$. The two styles are equally
  expressive: our judgment $\G \der \tau : v$ holds if and only if
  $(\vardiv \G v) \der \tau$ holds in Abel's system---but we found
  that this one extends more naturally to checking decomposability, as
  will later be necessary. The inference rules for the judgment
  $\G \der \tau : v$ are defined on Figure~\ref{fig/variance}.
\end{version}
\begin{version}{\Esop}
  We define a judgment to check the variance of a type
  expression. Given a context $\G$ of the form $\bar{v\alpha}$,
  that is, where each variable is annotated with a variance, the
  judgment $\G \der \tau : v$ checks that the
  expression $\tau$ varies along $v$ when the
  variables of $\tau$ vary along their variance in $\G$. 
For example,
  $(\vplus\alpha) \vdash \app{\tau}{\alpha} : \vplus$ holds when
  $\app{\tau}{\alpha}$ is covariant in its variable $\alpha$.  The
  inference rules for the judgment $\G \der \tau : v$ are defined
  on Figure~\ref{fig/variance}.

  The parameter $v$ evolves when going into subderivations: when
  checking $\G \der \tau_1 \to \tau_2 : v$, contravariance is
  expressed by checking
  $\G \der \tau_1 : \rel{\varcomp v \vminus}$. Previous work
  (on variance as \cite{polarized-subtyping-for-sized-types} and
  \cite{csharp-generalized-constraints}, but also on irrelevance as in
  \cite{pfenning:intextirr}) used no such parameter, but modified the
  context instead, checking $\vardiv \G \vminus \der \tau_1$ for
  some ``variance cancellation'' operation $v \vardiv w$
  (see \cite{polarized-subtyping-for-sized-types} for
  a principled presentation). Our own inference rules preserve the
  same context in the whole derivation and can be more easily adapted
  to the decomposability judgment $\G \der \tau : v \To v'$ that we
  introduce in \S\ref{sec/syntactic-decomposability}.
\end{version}

\paragraph{A semantics for variance assignment}

This syntactic judgment $\Gamma \der \tau : v$ corresponds to
a semantics property about the types and context involved, which
formalizes our intuition of ``when the variables vary along $\Gamma$,
the expression $\tau$ varies along $v$''. We also give a few formal
results about this judgment.

\begin{definition}[Interpretation of the variance checking judgment]
\label{def/vc/semantics}\indent\\
We write $\sem{\G \der \tau : v}$ for the property:
$
\forall \bs, \bs', \;
    \bs \prec_{\G} \bs'
    \implies \app \tau \bs \prec_v \app \tau {\bs'}
$.
\end{definition}
\begin{lemma}[Correctness of variance checking]
  \label{lem/variance-checking-correct}
  $\G \der \tau : v$ is provable if and only if
  $\sem{\G \der \tau : v}$ holds.
\end{lemma}
\begin{proof}{inline}
  \Case{Soundness: $\G \der \tau : v$ implies
    $\sem{\G \der \tau : v}$} By induction on the derivation. In
  the variable case this is direct. In the $\ty{\bs}{t}$ case, for
  $\br, \br'$ such that $\br \prec_\G \br'$, we get
  $\forall i, \app{\sigma_i}{\br} \prec_{\varcomp v {w_i}} \app{\sigma_i}{\br'}$
  by inductive hypothesis, which allows to conclude, by definition of
  variance composition, that
  $\app{(\ty{\bs}{t})}{\br} \prec_v \app{(\ty{\bs}{t})}{\br'}$.

  \Case{Completeness: $\sem{\G \der \tau : v}$ implies
    $\G \der \tau : v$} By induction on $\tau$; in the variable
  case this is again direct. In the $\ty{\bs}{t}$ case, given
  $\br \prec_\G \br'$ such that
  $\app{(\ty{\bs}{t})}{\br} \prec_v \app{(\ty{\bs}{t})}{\br'}$ we can
  deduce by inversion that for each variable $\alpha_i$ of variance
  $w_i$ in $\app{\tau}{\ba}$ we have
  $\app{\sigma_i}{\br} \prec_{\varcomp v {w_i}} \app{\sigma_i}{\br'}$,
  which allows us to inductively build the subderivations
  $\G \der \sigma_i : \varcomp v {w_i}$.
\qed
\end{proof}

\begin{lemma}[Monotonicity] 
\label{lem/monotonicity}
  If $\G \der \tau : v$ is provable and $\G \leq \G'$ then
  $\G' \der \tau : v$ is provable.
\end{lemma}
\begin{proof}{hide}
Obvious. 
\end{proof}
\begin{version}{\Not\Esop}
\begin{lemma}
  If $\G \der \tau : v$ and $\G' \der \tau : v$ both hold,
  then $(\G \vee \G') \der \tau : v$ also holds.
\end{lemma}
\begin{proof}{hide}
Obvious
\end{proof}
\begin{corollary}[Principality] \label{principality}
  For any type $\tau$ and any variance $v$, there exists a minimal context
  $\Delta$ such that $\Delta \der \tau : v$ holds. That is, for any other
  context $\G$ such that $\G \der \tau : v$, we have $\Delta \leq
  \G$.
\end{corollary}
\end{version}

\begin{version}{\Esop}
\begin{lemma}[Principality] \label{principality}
  For any type $\tau$ and any variance $v$, there exists a minimal context
  $\Delta$ such that $\Delta \der \tau : v$ holds. That is, for any other
  context $\G$ such that $\G \der \tau : v$, we have $\Delta \leq
  \G$.
\end{lemma}
\end{version}

\begin{version}{\Not\Esop}
\paragraph{Inversion of subtyping}

We have mentioned in \ref{subtyping_inversion} the inversion
properties of our subtyping relation. From
$\ty{\bs}{t} \leq \ty{\bs'}{t}$ we can deduce subtyping relations on
the type parameters $\sigma_i, \sigma'_i$. This can be generalized to
any type expression $\app{\tau}{\ba}$:

\begin{theorem}[Inversion] \label{thm/inversion} For any type
  $\app{\tau}{\ba}$, variance $v$, and type sequences $\bs$ and
  $\bs'$, the subtyping relation
  $\app{\tau}{\bs} \prec_v \app{\tau}{\bs'}$ holds if and only if the
  judgment $\G \der \tau : v$ holds for some context $\G$ such
  that $\bs \prec_\G \bs'$.
\end{theorem}
\begin{proof}{}
  The  reverse implication, 
  is a direct application of the soundness of the variance judgment.

  The direct implication is proved by induction on
  $\app{\tau}{\ba}$. The variable case is direct: if
  $\app{\alpha}{\sigma} \prec_v \app{\alpha}{\sigma'}$ holds then for
  $\G$ equal to $(v\alpha)$ we indeed have $v\alpha \der \alpha : v$
  and $\sigma \prec_\G \sigma'$.

  In the $\ty{\bt}{t}$ case, we have that
  $\app{(\ty{\bt}{t})}{\bs} \prec_v \app{(\ty{\bt}{t})}{\bs'}$. Suppose
  the variance of $\ty{\ba}{t}$ is $\bar{w\alpha}$: by inversion on
  the head type constructor $\tyc{t}$ we deduce that for each $i$,
  $\app{\tau_i}{\bs} \prec_{\varcomp v {w_i}} \app{\tau_i}{\bs'}$. Our
  induction hypothesis then gives us a family of contexts
  $(\G_i)_\iI$ such that for each $i$ we have
  $\G_i \der \tau_i : \varcomp v {w_i}$. Furthermore,
  $\bs \prec_{\G_i} \bs'$ holds for all $\G_i$, which means
  that $\bs \prec_{\mathop{\wedge}_\iI \G_i} bs'$. Let's define
  $\G$ as $\mathop{\wedge}_\iI \G_i$. By construction we have
  $\G \geq \G_i$, so by monotonicity (Lemma~\ref{lem/monotonicity}) we
  have $\G \der \tau_i : \varcomp v {w_i}$ for each $i$. This
  allows us to conclude $\G \der \ty{\bt}{t} : v$ as desired.
\qed
\end{proof}

Note that this would work even for type constructors that are not
$v$-closed: we are not comparing a $\app{\tau}{\bs}$ to any type
$\tau'$, but to a type $\app{\tau}{\bs'}$ sharing the same
structure---the head constructors are always the same.

For any given pair $\bs, \bs'$ such that
$\app{\tau}{\bs} \prec_v \app{\tau}{\bs'}$ we can produce a context
$\G$ such that $\bs \prec_\G \bs'$. But is there a common
context that would work for any pair? Indeed, that is the lowest
possible context, the principal context $\G$ such that
$\G \der \tau : v$.

\begin{corollary}[Principal inversion] 
\label{cor/principal-inversion}
  If $\Delta$ is principal for $\Delta \der \tau : v$, then for any
  type sequences $\bs$ and $\bs'$, the subtyping relation
  $\app{\tau}{\bs} \prec_v \app{\tau}{\bs'}$ implies
  $\bs \prec_\Delta \bs'$.
\end{corollary}
\begin{proof}{}
  Let $\Delta$ be the principal context such that
  $\Delta \der \tau : v$ holds. For any $\bs, \bs'$ such that
  $\app{\tau}{\bs} \prec_v \app{\tau}{\bs'}$, by inversion
  (Theorem~\ref{thm/inversion}) we have some $\G$ such that
  $\G \der \tau : v$ and $\bs \prec_\G \bs'$. By definition of $\Delta$,
  $\Delta \leq \G$ so $\bs \prec_\Delta \bs'$ also holds. That is,
  $\bs \prec_\Delta \bs'$ holds for any $\bs, \bs'$ such that
  $\app{\tau}{\bs} \prec_v \app{\tau}{\bs'}$.
\qed
\end{proof}
\end{version}

\begin{version}{\Esop}
  We can generalize inversion of head type
  constructors~(\S\ref{subtyping_inversion}) to whole type
  expressions. The most general inversion is given by the principal
  context.

\begin{theorem}[Inversion]
  \label{thm/inversion}
  \label{cor/principal-inversion}

  For any type $\app{\tau}{\ba}$, variance $v$, and type sequences
  $\bs$ and $\bs'$, the subtyping relation
  $\app{\tau}{\bs} \prec_v \app{\tau}{\bs'}$ holds if and only if the
  judgment $\G \der \tau : v$ holds for some context $\G$ such
  that $\bs \prec_\G \bs'$.

  Furthermore, if $\app{\tau}{\bs} \prec_v \app{\tau}{\bs'}$, then
  $\bs \prec_\Delta \bs'$ holds, where $\Delta$ is the more general
  context such that $\Delta \der \tau : v$ holds.
\end{theorem}
\end{version}

\paragraph{Checking variance of type definitions}

We have all the machinery in place to explain the checking of
ADT variance declarations.  The well-signedness criterion of Simonet
and Pottier gives us a general semantic characterization of which
definitions are correct: a definition 
$\mc{type}~\ty{\Gva}{t} =
(\mc{K}_c~\mc{of}~(\app{\tau^c}{\ba})_{c \in C}$ 
is correct if, for each constructor $c$, we have: 
\[ \forall \bs, \forall \bs', \quad
   \ty{\bs}{t} \leq \ty{\bs'}{t} \implies 
     \app{\tau}{\bs} \leq \app{\tau}{\bs'}
\]
By inversion of subtyping, $\ty{\bs}{t} \leq \ty{\bs'}{t}$ implies
$\sigma_i \prec_{v_i} \sigma'_i$ for all $i$. 
Therefore, it suffices to check that:
\[ \forall \bs, \forall \bs', \quad
   (\forall i, \sigma_i \prec_{v_i} \sigma'_i)
     \implies
     \app{\tau}{\bs} \leq \app{\tau}{\bs'}
\]
This is exactly the semantic property corresponding to the
judgment ${\Gva\der\tau:\rel\vplus}$! 
That is, we have reduced  soundness verification of an algebraic type definition
to a mechanical syntactic check on the constructor argument type.

This syntactic criterion is very close to the one implemented in
actual type checkers, which do not need to decide general subtyping
judgments---or worse solve general subtyping constraints---to check
variance of datatype parameters.  Our aim is now to find a similar
syntactic criterion for the soundness of variance annotations on
guarded algebraic datatypes, rather than simple algebraic datatypes.

\subsection{Variance annotations in \GADTs}

\paragraph {A general description of \GADTs}

When used to build terms of type $\ty{\ba}{t}$, a constructor
$\mc{K}~\mc{of}~\tau$ behaves like a function of type
$\forall \ba. (\tau \to \ty{\ba}{t})$. Remark that the codomain is
exactly $\ty{\ba}{t}$, the type $\tyc{t}$ instantiated with parametric
variables. \GADTs arise by relaxing this restriction, allowing to
specify constructors with richer types of the form
$\forall \ba. (\tau \to \ty{\bs}{t})$. See for example the declaration
of constructor \code{Prod} in the introduction:
\begin{lstlisting}[xleftmargin=4em]
  | Prod : $\forall \beta \gamma .\ \ty{\beta}{expr} * \ty{\gamma}{expr}
                                  \to \ty{(\beta*\gamma)}{expr}$
\end{lstlisting}
Instead of being just $\ty{\alpha}{expr}$, the codomain is now
$\ty{(\beta*\gamma)}{expr}$. We moved from simple algebraic datatypes
to so-called \emph{generalized} algebraic datatypes. This approach is
natural and convenient for the users, so it is exactly the syntax
chosen in languages with explicit \GADTs support, such as Haskell and
OCaml, and is reminiscent of the inductive datatype definitions of
dependently typed languages.

\begin{version}{\Not\Esop}
However, for formal study of \GADTs, a different formulation based on
equality constraints is preferred. The idea is that we will force
again the codomain to be exactly $\ty{\alpha}{expr}$, but allow
additional type equations such as
$\alpha = \beta * \gamma$ in this example:
\begin{lstlisting}[xleftmargin=4em]
| Prod : $\forall \alpha.\ \forall \beta \gamma [ \alpha = \beta * \gamma ].\ 
            \ty{\beta}{expr} * \ty{\gamma}{expr}
            \to \ty{\alpha}{expr}$
\end{lstlisting}
This restricted form justifies the name of \emph{guarded} algebraic
datatype. The $\forall\beta\gamma[D].\tau$ notation, coming from Simonet
and Pottier, is a \emph{constrained type scheme}: $\beta,\gamma$ may
only be instantiated with type parameters respecting the constraint
$D$. Note that, as $\beta$ and $\gamma$ do not appear in the codomain
anymore, we may equivalently change the outer universal into an
existential on the left-hand side of the arrow:
\begin{lstlisting}[xleftmargin=4em]
| Prod : $\forall \alpha.\ (\exists \beta \gamma [ \alpha = \beta * \gamma ].\ 
            \ty{\beta}{expr} * \ty{\gamma}{expr})
            \to \ty{\alpha}{expr}$
\end{lstlisting}
In the general case, a \GADT definition for $\ty{\ba}{t}$ is composed of a set
of constructor declarations, each of the form:
\begin{lstlisting}[xleftmargin=4em]
| K : $\forall \ba.\
         (\exists \bb [ \ba = \app \bs \bb ] .\,\app \tau {\ba,\bb})
         \to \ty{\ba}{t}$
\end{lstlisting}
or, reusing the classic notation,
\begin{lstlisting}[xleftmargin=4em]
| K of $\exists \bb [ \ba = \app \bs \bb ] .\,
                            \app \tau {\ba,\bb}$
\end{lstlisting}
Without loss of generality, we can conveniently assume that the
variables $\ba$ do not appear in the parameter type $\tau$
anymore: if some $\alpha_i$ appears in $\tau$, one may always pick
a fresh existential variable $\beta$, add the constraint
$\alpha = \beta$ to $D$, and consider $\tau[\beta/\alpha]$. Let us
re-express the introductory example in this form, that is, 
$\mc{K of}~\exists\bb[\ba = \app \bs \bb].\, \app \tau \bb$:
\end{version}

\begin{version}{\Esop}
  However, for formal study of \GADTs, a different formulation based
  on equality constraints is preferred. We will use the following
  equivalent presentation, in the syntax of Simonet and Pottier. The
  idea is that instead of having $\ty{(\beta * \gamma)}{t}$ as
  codomain of the constructor \code{Prod}, we will force it to be
  $\ty{\alpha}{t}$ again, by adding an explicit type equality $\alpha
  = \beta * \gamma$.  \pagebreak 
\end{version}

\begin{lstlisting}[xleftmargin=4em]
type $\ty{\alpha}{expr}$ =
  | Val of $\exists \beta [\alpha = \beta].\,\beta$
  | Int of $[\alpha = \tyc{int}].\,\tyc{int}$
  | Thunk of $\exists \beta \gamma [\alpha = \gamma].\, \ty{\beta}{expr} * (\beta \to \gamma)$
  | Prod of $\exists \beta\gamma[\alpha = \beta*\gamma].\, \ty{\beta}{expr} * \ty{\gamma}{expr}$
\end{lstlisting}

\begin{version}{\Not\Esop}
If all constraints between brackets are of the simple form
$\alpha_i = \beta_i$ (for distinct variables $\alpha_i$ and $\beta_i$), as for
the constructor \code{Thunk}, then we have a constructor with
existential types as described by Laüfer and Odersky
\cite{laufer-odersky-92}. If furthermore there are no other
existential variables than those equated with a type parameter, as in
the \code{Val} case, we have an usual algebraic type constructor; of
course the whole type is ``simply algebraic'' only if each of its
constructors is algebraic.
\end{version}

In the rest of the paper, we extend our former core language with such
guarded algebraic datatypes. This impacts the typing rules (which are
precisely defined in Simonet and Pottier), but not the notion of
subtyping, which is defined on (\GADT) type constructors with variance
$\mc{type}~\ty{\Gva}{t}$ just as it previously was on simple
datatypes. What needs to be changed, however, is the soundness
criterion for checking the variance of type definitions.

\paragraph{The correctness criterion}

Simonet and Pottier~\cite{simonet-pottier-hmg-toplas} define a general
framework HMG(X) to study type systems 
with \GADTs where the type equalities in bounded quantification are
generalized to an arbitrary constraint language. They make few
assumptions on the type system used, mostly that it has function types
$\sigma \to \tau$, user-definable (guarded) algebraic datatypes
$\ty{\ba}{t}$, and a subtyping relation $\sigma \leq \tau$ (which
may be just equality, in languages without subtyping).

They use this general type system to give static semantics
(typing rules) to a fixed untyped lambda-calculus equipped with
datatype construction and pattern matching operations. They are able
to prove a type soundness result under just some general assumptions
on the particular subtyping relation $\rel\leq$. Here are the three
requirements to get their soundness result:
\begin{enumerate}
\item Incomparability of distinct types: for all types $\tau_1,
  \tau_2, \bs, \bs'$ and distinct datatypes $\ty{\ba}{t},
  \ty{\ba'}{t'}$, the types $(\tau_1 \to \tau_2)$, $\tau_1 * \tau_2$,
  $\ty{\bs}{t}$ and $\ty{\bs'}{t'}$ must be pairwise incomparable
  (both $\nleq$ and $\ngeq$) --- this is where our restriction to an
  atomic subtyping relation, discussed
  in~\S\ref{atomic_subtyping_restriction}, comes from.

\item Decomposability of function and product types: if
  $\tau_1\to\tau_2 \leq \sigma_1\to\sigma_2$
  (respectively $\tau_1 * \tau_2 \leq \sigma_1 * \sigma_2$), we must
  have $\tau_1 \geq \sigma_1$ (resp. $\tau_1 \leq \sigma_1$) and
  $\tau_2 \leq \sigma_2$.

\item Decomposability of datatypes\footnote{This is an extended
    version of the soundness requirement for algebraic datatypes: it
    is now formulated in terms guarded existential types
    $\exists\bb[D]\tau$ rather than simple argument types $\tau$.}:
  for each datatype $\ty{\ba}{t}$ and all type vectors $\bs$ and $\bs'$
  such that $\ty{\bs}{t} \leq \ty{\bs'}{t}$, we must have
  $(\exists \bb[\app D \bs]\tau) \leq (\exists \bb[\app D {\bs'}] \tau)$
  for each  constructor
  $\mc{K of}~\exists \bb[\app D {\bb,\ba}].\,\app \tau \bb$.
\end{enumerate}
\XXX[GS]{Didier, please review the paragraph below: it was modified to
  mention and explain the ``variance taken granted'' approach of later
  proofs.\\}
Those three criteria are necessary for the soundness proof. We will
now explain how variance of type parameters impact those requirements,
that is, how to match a \GADT implementation against a variance
specification. With our definition of subtyping based on variance, and
the assumption that the datatype $\ty{\bar{v\alpha}}{t}$ we are
defining indeed has variance $\bar{v\alpha}$, is the \GADT
decomposability requirement (item 3 above) satisfied by all its
constructors? If so, then the datatype definition is sound and can be
accepted. Otherwise, the datatype definition does not match the
specified variance, and should be rejected by the type checker.


\section{Checking variances of \GADT}
\label{sec:checking_variance_of_gadts}


For every type definition, we need to check that the decomposability
requirement of Simonet and Pottier holds. Remark that it is expressed
for each \GADT constructor independently of the other constructors for
the same type: we can check one constructor at a time. 

Assume we
check a fixed constructor $\mc{K}$ of argument type
$\exists \bb[\app D {\ba,\bb}].\,\app \tau \bb$.
Simonet and Pottier prove that their requirement is equivalent to the
following formula, which is more convenient to manipulate:
$$
\forall \bs, \bs', \br, \uad
\bigparens{
   \ty{\bs}{t} \leq \ty{\bs'}{t}
     \wedge \app D{\bs,\br}
   \implies
   \exists \br',\; \app D{\bs',\br'}
     \wedge \app\tau\br \leq \app\tau{\br'}
}
\eqno 
(\DefRule {req-SP})
$$
The purpose of this section is to extract a practical criterion
equivalent to this requirement. It should not be expressed as
a general constraint satisfaction problem, but rather as
a syntax-directed and decidable algorithm that can be used in
a type-checker---without having to implement a full-blown constraint
solver.

\begin{version}{\Not\Esop}
\paragraph{A remark on the non-completeness}
\label{simonet-pottier-not-complete}

Note that while the criterion \Rule {req-SP} is \emph{sound},
it is not \emph{complete}---even in the simple ADT case.

For a constructor $\mc{K}~\mc{of}~\tau$ of $\ty{\bs}{t}$,
the justification for the fact that, under the hypothesis
$\ty{\bs}{t} \leq \ty{\bs'}{t}$, we should have
$\app{\tau}{\bs} \leq \app{\tau}{\bs'}$ is the following: given
a value $v$ of type $\app{\tau}{\bs}$, we can build the value
$\mc{K}~v$ at type $\ty{\bs}{t}$, and coerce it to $\ty{\bs'}{t}$. We
can then deconstruct this value by matching the constructor $K$, whose
argument is of type $\app{\tau}{\bs'}$. But this whole computation,
$(\mc{match}~(\mc{K}~v~\mc{:>}~\ty{\bs'}{t})~\mc{with}~\mc{K}~{x}~\to~x)$,
reduces to $v$, so for value reduction to hold we need to also have
$v : \app{\tau}{\bs'}$.

For this whole argument to work we need a value at type
$\app{\tau}{\bs}$. In fact, if the type $\app{\tau}{\bs}$ is not
inhabited, it can fail to satisfy \Rule {req-SP}
and still be sound: this criterion is not complete. See the following
example in a consistent system with an uninhabited type $\bot$:
\begin{lstlisting}
type $\ty{+\alpha}{t}$ =
  | T of $\tyc{int}$
  | Empty of $\bot * (\alpha \to \tyc{bool})$
\end{lstlisting}
Despite $\alpha$ occurring in a contravariant position in the dead
\code{Empty} branch (which violates the soundness criterion of Simonet
and Pottier), under the assumption that the $\bot$ type really is
uninhabited we know that this \code{Empty} constructor will never be
used in closed code, and the contravariant occurrence will therefore
never make a program ``go wrong''. The definition is correct, yet
rejected by \Rule {req-SP}, which is therefore incomplete.

Deciding type inhabitation in the general case is a very complex
question, which is mostly orthogonal to the presence and design of
\GADT in the type system. There is, however, one clear interaction
between the type inhabitation question and \GADTs. If a \GADT
$\ty{\ba}{t}$ is instantiated with type variables $\bs$ that satisfy
none of the constraints $\app{D}{\alpha}$ of its constructors
$\mc{K}~\mc{of}~\exists\bb[\bar{D}].\tau$, then we know that
$\ty{\bs}{t}$ is not inhabited. This is related to the idea of
``domain information'' that we discuss in the Future Work
section~(\S\ref{future_work:domain_information}).
\end{version}

\subsection{Expressing decomposability}
\label{sec/expressing_decomposability}

If we specialize \Rule {req-SP} to the \code{Prod} constructor
of the $\ty{\alpha}{expr}$ example datatype, \ie. 
$\mc{Prod}~\mc{of}~\exists \beta \gamma[\alpha = \beta * \gamma]
\ty{\beta}{expr} * \ty{\gamma}{expr}$, we get:
\[
\begin{array}{l}
  \forall \sigma, \sigma', \rho_1, \rho_2, \\\uad
  \bigparens{
  \ty{\sigma}{expr} \leq \ty{\sigma'}{expr}
    \wedge
    \sigma = \rho_1 * \rho_2
  \implies
  \exists \rho'_1, \rho'_2, \parens {\sigma' = \rho'_1 * \rho'_2
    \wedge
    \rho_1 * \rho_2 \leq \rho'_1 * \rho'_2}
  }
\end{array}
\]
We can substitute equalities and use the (assumed)\XXX{This step
  should be clearer: variances can be assumed and the criterion says
  whether the consequences of these assumptions are sound. I think it
  should be discussed earlier at the ``toplevel'' of the
  paper.\\}
\XXX[GS]{Can it be left as is with the new paragraph
  mentioning this before? We don't really have the space budget to be
  more explicit here.} covariance to simplify the subtyping constraint
$\ty{\sigma}{expr} \leq \ty{\sigma'}{expr}$ into $\sigma \leq
\sigma'$:
\[
  \forall \sigma', \rho_1, \rho_2,\uad
  \bigparens
    {\rho_1 * \rho_2 \leq \sigma'
  \implies
  \exists \rho'_1, \rho'_2, \uad  
    \parens 
      {\sigma' = \rho'_1 * \rho'_2
       \wide\wedge
       \rho_1 \leq \rho'_1
       \wide\wedge
       \rho_2 \leq \rho'_2
      }
  }
\eqno \llabel 1
\]
This is the \emph{upward closure} property mentioned in the
introduction. This transformation is safe only if any supertype
$\sigma'$ of a product $\rho_1 * \rho_2$ is itself a product, \ie. is
of the form $\rho'_1 * \rho'_2$ for some $\rho'_1$ and $\rho'_2$.

More generally, for a type $\G \der \sigma$ and a variance $v$, we
are interested in a closure property of the form \[
  \forall (\br : \G), \sigma',\quad
    \app \sigma \br \prec_{v} \sigma' \implies
    \exists (\br' : \G),\,\sigma' = \app \sigma {\br'}
\]
Here, the context $\G$ represents the set of existential variables
of the constructor ($\beta$ and $\gamma$ in our example). We can
easily express the condition $\rho_1 \leq \rho'_1$ and
$\rho_2 \leq \rho'_2$ on the right-hand side of the implication by
considering a context $\G$ annotated with variances
$(\vplus\beta, \vplus\gamma)$, and using the context ordering
$\rel{\prec_\G}$. Then, \lref 1 is equivalent to: 
$$
  \forall (\br : \G), \sigma',\quad
    \app \sigma \br \prec_{v} \sigma' \implies
    \exists (\br' : \G),\uad\br \prec_\G \br'
      \wedge \sigma' = \app \sigma {\br'}
$$
Our aim is now to find a set of inference rules to check decomposability; we
will later reconnect it  to \Rule {req-SP}.
In fact, we study a slightly more general relation, where the
equality $\app\sigma{\br'}~=~\sigma'$ on the right-hand side is
relaxed to an arbitrary relation
$\app{\sigma}{\br'} \prec_{v'} \sigma'$:

\begin{definition}[Decomposability]
\label{def/decompasability}
Given a context $\G$, a type expression $\app{\sigma}{\bb}$ and two
variances $v$ and $v'$, we say that $\sigma$ is \emph{decomposable} under
$\G$ from variance $v$ to variance $v'$, which we write $\G \Der \sigma : v
\vto v'$, if the property
$$
\forall (\br : \G), \sigma',\uad
      \app \sigma \br \prec_{v} \sigma' \implies
      \exists (\br' : \G),\uad
        \br \prec_\G \br'
        \wide\wedge \app \sigma {\br'} \prec_{v'} \sigma'
$$
holds.
\end{definition}
We use the symbol $\Der$ rather than $\der$ to highlight the fact that
this is just a logic formula, not the semantics criterion
corresponding to an inductive judgment, nor a syntactic judgment---we
will introduce one later in section
\ref{sec/syntactic-decomposability}.

Remark that, due to the \emph{positive} occurrence of the relation
$\prec_\G$ in the proposition ${\G \Der \tau : v \vto v'}$ and the
anti-monotonicity of $\prec_\G$, this formula is ``anti-monotonous''
with respect to the context ordering $\G \leq \G'$. This corresponds
to saying that we can still decompose, but with less information on
the existential witness $\br'$.

\begin{lemma}[Anti-monotonicity]
\label{lem/anti-monotonicity}
If $\G \Der \tau : v \vto v'$ holds
and $\G' \leq \G$, then $\G' \Der \tau : v \vto v'$ also holds.
\end{lemma}
\begin{version}{\Not\Esop}
Our final decomposability criterion, given below in Figure~\ref{fig/decomposability},
requires both correct variances and a decomposability property, so it
will be neither monotonous nor anti-monotonous with respect to the
context argument.

\end{version}
In the following subsections, we study the subtleties of
decomposability. 

\subsection{Variable occurrences}

In the \code{Prod} case, the type whose decomposability was considered
is $\beta * \gamma$ (in the context $\beta, \gamma$).  In this very
simple case, decomposability depends only on the type constructor for
the product. In the present type system, with very strong
invertibility principles on the subtyping relation, both upward and
downward closures hold for products---and any other head type
constructor. In the general case, we require that this specific type
constructor be upward-closed.

In the general case, the closure of the head type constructor alone is
not enough to ensure decomposability of the whole type. For example,
in a complex type expression with subterms, we should consider the
closure of the type constructors appearing in the subterms as
well. Besides, there are subtleties when a variable occurs several
times.

For example, while $\beta * \gamma$ is decomposable from $\rel\vplus$
to $\rel\veq$, $\beta * \beta$ is not: $\bot * \bot$ is an
instantiation of $\beta * \beta$, and a subtype of, \eg., $\tyc{int} *
\tyc{bool}$, but it is not equal to $\app{(\beta * \beta)}{\gamma'}$
for any $\gamma'$. The same variable occurring twice in covariant
position (or having one covariant and one invariant or contravariant
occurence) breaks decomposability.

On the other hand, two invariant occurrences are possible:
$\ty{\beta}{ref} * \ty{\beta}{ref}$ is upward-closed (assuming the
type constructor $\tyc{ref}$ is invariant and upward-closed): if
$(\ty{\sigma}{ref} * \ty{\sigma}{ref}) \leq \sigma'$, then by
upward closure of the product, $\sigma'$ is of the form
$\sigma'_1 * \sigma'_2$, and by its covariance
$\ty{\sigma}{ref} \leq \sigma'_1$ and
$\ty{\sigma}{ref} \leq \sigma'_2$. Now by invariance of $\tyc{ref}$ we
have $\sigma'_1 = \ty{\sigma}{ref} = \sigma'_2$, and therefore
$\sigma'$ is equal to $\ty{\sigma}{ref} * \ty{\sigma}{ref}$, which is
an instance\footnote {We use the term \emph{instance} to denote 
the replacement  of all the free variables of a type expression under
context by   closed types---not the specialization of an ML type scheme.}
 of $\ty{\beta}{ref} * \ty{\beta}{ref}$. 

Finally, a variable may appear in irrelevant positions without
affecting closure properties; $\beta * (\ty{\beta}{irr})$
(where $\tyc{irr}$ is an upward-closed irrelevant type, defined for
example as $\mc{type}~\ty{\alpha}{irr} = \tyc{int}$) is upward closed:
if $\sigma * (\ty{\sigma}{irr}) \leq \sigma'$, then $\sigma'$ is of
the form $\sigma'_1 * (\ty{\sigma'_2}{irr})$ with
$\sigma \leq \sigma'_1$ and $\sigma \Join \sigma'_2$, which is
equiconvertible to $\sigma'_1 * (\ty{\sigma'_1}{irr})$ by irrelevance,
an instance of $\beta * (\ty{\beta}{irr})$.

\XXX[DR]{I don't see why you need equivalence and not just subtyping 
in the example above.}

\XXX[GS]{To prove that $\beta * (\ty{\beta}{irr})$ is
  $(\vplus \To \veq)$-decomposable, we need to prove that for any
  $\rho$ and $\sigma'$ such that
  $\rho * (\ty{\rho}{irr}) \leq \sigma'$ we have some $\rho'$ such
  that $\sigma'$ is \emph{equal} (or equiconvertible, now defined in
  the paragraph insisting on our notion of equality in the definition
  of the subtyping relation) to $\rho' * (\ty{\rho'}{irr})$.}

\subsection{Context zipping}

The intuition to think about these different cases is to consider
that, for any $\sigma'$, we are looking for a way to construct a ``witness''
$\bs'$ such that $\app \tau {\bs'} = \sigma'$ 
from the hypothesis $\app \tau \bs \prec_v \sigma'$.
When a type variable appears only once, its witness can be determined
by inspecting the corresponding position in the type $\sigma'$. For
example in $\alpha * \beta \leq \tbool * \tint$, the mapping $\alpha
\mapsto \tbool, \beta \mapsto \tint$ gives the witness pair $\tbool,
\tint$.

However, when a variable appears twice, the two witnesses
corresponding to the two occurrences may not coincide. (Consider for
example $\beta * \beta \leq \tyc{bool} * \tyc{int}$.) If a variable
$\beta_i$ appears in several \emph{invariant} occurrences, the witness
of each occurrence is forced to be equal to the corresponding subterm
of $\app \tau \bs$, that is $\sigma_i$, and therefore the various
witnesses are themselves equal, hence compatible.  On the contrary,
for two covariant occurrences (as in the $\beta * \beta$ case), it is
possible to pick a $\sigma'$ such that the two witnesses are
incompatible---and similarly for one covariant and one invariant
occurrence. Finally, an irrelevant occurrence will never break closure
properties, as all witnesses (forced by another occurrence) are
compatible.

To express these merging properties, we define a ``zip''\footnote {The
  idea of context merging and the term ``zipping'' are inspired by
  Montagu and Remy~\cite{montagu-remy-09}} operation $v_1 \zip v_2$, that
formally expresses which combinations of variances are possible for
several occurrences of the same variable; it is a partial operation
(for example, it is not defined in the covariant-covariant case, which
breaks the closure properties) with the following table:
$$
\def \C#1{\multicolumn{1}{C}{#1}}
\begin{tabular}{R|c@{}|C|C|C|C|l}
  v \zip w && \C{\veq} & \C{\vplus} & \C{\vminus} & \C{\virr} & w 
\\ 
\cline{1-1}\cline{3-7}\noalign{\vskip \doublerulesep}\cline{1-1}\cline{3-6}
  \veq     && \veq   &          &           & \veq    \\ \cline{3-6}
  \vplus   &&        &          &           & \vplus  \\ \cline{3-6}
  \vminus  &&        &          &           & \vminus \\ \cline{3-6}
  \virr    && \veq   & \vplus   & \vminus   & \virr   \\ \cline{3-6}
  v        &\multicolumn{5}{C}{}
\end{tabular}
$$

\begin{version}{\Not\Esop}
The following lemma uses zipping to merge together the results of the
decomposition of several subterms $(\sigma_i)_i$ into a ``simultaneous
decomposition''.

\begin{definition}[Simultaneous decomposition]
\label{def/simultaneous-decomposition}
Given a context $\G$, and families of type expressions $(\sigma_i)_\iI$ and
variances $(v_i)_\iI$ and $(v'_i)_\iI$, we define the following
``simultaneous closure property'' ${\G \Der (T_i : v_i \vto v'_i)_\iI}$
defined as :
$$
    \forall (\br : \G), \bs',\uad
      \parens 
        {\forall \iI, \app{\sigma_i}{\br} \prec_{v_i} \sigma'_i}
      \implies
      \exists (\br' : \G),\uad \br \prec_\G \br'
        \wedge \parens 
         {\forall \iI, \app {\sigma_i} {\br'} \prec_{v'_i} \sigma'_i}
\nobelowdisplayskip
$$
\end{definition}
\begin{lemma}[Soundness of zipping]
\label{lem/zip-soundness}
Suppose we have families of type expressions $(\app{T_i}{\bb})_\iI$, contexts
$(\G_i)_\iI$ and variances $(v_i)_\iI$ and $(v'_i)_\iI$ such that
$\mathop{\zip}_\iI \G_i$ exists and for all $i$ we have both
$\G_i \der T_i : v_i$ and 
$\G_i \Der T_i : v_i \vto v'_i$.
Then, we have $(\mathop{\zip}_\iI \G_i) \Der (\sigma_i : v_i \vto v'_i)_\iI$.
\end{lemma}

\begin{proof}{}
\locallabelreset
Without loss of generality, 
\XXX[TODO]{Say why...}
we can consider that there are only two
type expressions $\app{T_1}{\bb}$ and $\app{T_2}{\bb}$, and that the
free variables $\bb$ is reduced to  a single variable $\beta$. 
Let $w_1, w_2$ be the respective variances of $\beta$ in $\G_1,
\G_2$. 
We know that $(w_1 \zip w_2)$ exists  and is equal to the variance $w$
of the variable $\beta$ in $\G_1 \zip \G_2$. 

\let \br \rho

Our further assumptions are 
$\G_i \der T_i : v_i$~\llabel 1 and 
$\G_i \Der T_i : v_i \vto v'_i$~\llabel 2 
for $i$ in $\set {1,2}$. 
The  expansion of \lref 2 is:
$$
  \forall \iI,\, \forall \br,\sigma'_i,\quad
     \app{T_i}{\br} \prec_{v_i} \sigma'_i
     \wide\implies \exists \br', \uad \br \prec_{\G_i} \br'
       \;\wedge\;\app{T_i}{\br'} \prec_{v'_i} \sigma'_i
\eqno \llabel 3
$$
Our goal is to prove  $\G_1 \zip \G_2 \Der (\sigma_i : v_i\vto
v'_i)_\iI$, which is equivalent to:  
$$
  \forall \br,\bs',\quad
    (\forall \iI, (\app{T_i}{\br} \prec_{v_i} \sigma'_i))
    \wide\implies \exists \br',\uad
        \br \prec_{\G} \br' \wide\wedge
        \forall \iI, \app{T_i}{\br'} \prec_{v'_i} \sigma'_i
\eqno \llabel C
$$
Assume given $\rho$ and $(\sigma'_1, \sigma'_2)$ such that
$\app{T_1}{\br} \prec_{v_1} \sigma'_1$~\llabel {H1} and
$\app{T_2}{\br} \prec_{v_2} \sigma'_2$~\llabel {H2}. 
Applying~\lref 3 with $i$ equal to $1$ and \lref {H1} ensures the existence of 
a $\rho'_1$ such that $\rho \prec_{w_1} \rho'_1$~\llabel {rw1} and
$\app{T_1}{\rho'_1} \prec_{v_1} \sigma'_1$~\llabel {rp1}. 
Similarly, there exists  $\rho'_2$ such that
$\rho \prec_{w_2} \rho'_2$~\llabel {rw2} and
$\app{T_2}{\rho'_2} \prec_{v_2} \sigma'_2$~\llabel {rp2}. 
To establish \lref C, it remains to build a single $\rho'$ that satisfies
$\rho \prec_w \rho'$~\llabel {r0}, $\app{T_1}{\rho'} \prec_{v_1}
\sigma'_1$~\llabel {r1} and $\app{T_2}{\rho'} \prec_{v_2} \sigma'_2$~\llabel
{r2}, simultaneously.
We reason by case analysis on $w_1$ and $w_2$ (restricted to the cases where
the zip exists).

If both $w_1$ and $w_2$ are $\rel\virr$, we take either $\rho'_1$ or
$\rho'_2$ for: since $w$ is $\virr$, we \lref {r0} trivially holds; 
Furthermore, $\app{T}{\rho'} = \app{T}{\rho'_1} = \app{T}{\rho'_2}$ by 
irrelevance of $v_i$ and \lref 1; therefore \lref {r1} and \lref {r2} 
follow from \lref {rp1} and \lref {rp2}.

If only one of the $w_i$ is $\rel\virr$, we'll suppose that it is $w_1$. We
then take $\rho_2$ for $\rho'$.  Since $\virr_1 \zip w_2$ is $w_2$, \lref
{r0} follows from \lref {rw2}; Furthermore, $\app{T}{\rho'} =
\app{T}{\rho'_1}$ by irrelevance of $v_1$ and \lref 1 while
$\app{T_2}{\rho'} = \app{T_2}{\rho'_2}$ holds by construction; Hence, as in
the previous case, \lref {r1} and \lref {r2} follow from \lref {rp1} and
\lref {rp2}.

Finally, if both $w_1$ and $w_2$ are $\rel\veq$, then 
\lref {rw1} and \lref {rw2} implies $\rho'_1 = \rho = \rho'_2$.
We take $\rho$ for $\rho'$ and all three conditions are obviously satisfied.
\qed
\end{proof}
This lemma admits a kind of converse lemma stating completeness of
zipping, that says that if ${\G \Der (T_i : v_i \vto v'_i)_i}$ holds, then
$\G$ is indeed related to a zip of contexts $(\G_i)_i$ that pairwise
decompose each of the $(T_i)_i$. However, the proof of completeness is more
delicate and
\Esop{can only be found in the extended version}
{we prove it separately in \S\ref{sec/zip-completeness}}.
\end{version}

\subsection{Syntactic decomposability}

\label {sec/syntactic-decomposability}

\begin{mathparfig}{fig/decomposability}{Syntactic decomposablity}
\inferrule[sc-Triv]
  {v \geq v'\\\G \der \tau : v}
  {\G \der \tau : v \To v'}

\inferrule[sc-Var]
  {w\alpha \in \G \\ w = v}
  {\G \der \alpha : v \To v'}

\inferrule[sc-Constr]
  {\G \der \mc{type}~\ty{\bar{w\alpha}}{t} : v\textrm{-closed}\\
   \G = \mathop{\zip_i} \G_i\\
   \forall i,\ \G_i \der \sigma_i : \varcomp v {w_i} \To \varcomp{v'}{w_i}}
  {\G \der \ty{\bar{\sigma}}{t} : v \To v'}
\end{mathparfig}

Equipped with the zipping operation, 
we introduce a judgment $\G \der \tau : v \To v'$ to express 
decomposability, syntactically, defined by the inference rules on
Figure~\ref{fig/decomposability}. 


We
\Esop
{sometimes need to} {rely on the zip soundness
(Lemma~\ref{lem/zip-soundness}) to} merge sub-derivations into larger ones,
so in addition to decomposability, the judgments simultaneously 
ensures that $v$ is a correct variance for $\tau$ under $\G$.
\begin{version}{\False}
  We want a judgment that captures decomposability. To be able to
  combine different subderivations together using our $\zip$
  operation, we will need an additional hypothesis of variance
  checking (see the details in our extended version).
\end{version}
Actually, in order to understand the details of this judgment, 
it is quite instructive to
compare it with the variance-checking judgment $\G \der \tau : v$ defined on
Figure~\ref{fig/variance}.

\begin{version}{\Esop}
  The first thing to notice is that the present rules are not
  completely syntax-directed: we first check whether $v \geq v'$
  holds, and if not, we apply syntax-directed inference rules;
  existence of derivations is still easily decidable. If $v \geq v'$
  holds, satisfying the semantics criterion is trivial: $\app \tau \bs
  \prec_v \tau'$ implies $\app \tau \bs \prec_{v'} \tau'$, so taking
  $\bs$ for $\bs'$ is always a correct witness, which is represented
  by Rule \Rule{sc-Triv}.  The other rules then follow the same
  structure as the variance-checking judgment.
\end{version}

Rule \Rule{sc-Var} is very similar to \Rule {vc-Var}, except that the
condition $w \geq v$ is replaced by a stronger equality $w = v$.
\begin{version}{\Not\Esop}
  The reason why the variance-checking judgment has an inequality
  $w \geq v$ is to make it monotonous in the environment---as
  requested by its corresponding semantics
  criterion~(Definition \ref{def/vc/semantics}). Therefore, the
  condition $w \geq v$ is necessary for completeness---and admissible.
  On the contrary, the present judgment ensures, according the
  semantic criterion~(Definition~\ref{def/cc/semantics}), that both
  the variance is correct (monotonous in the environment) and the type
  is decomposable, a property which is \emph{anti-monotonic} in the
  environment~(Lemma~\ref{lem/anti-monotonicity}).  Therefore, the
  semantics criterion $\sem{\G \der \tau : v \To v'}$ is invariant
  in $\G$ and, correspondingly, the variable rule must use a strict
  equality.
\end{version}
\begin{version}{\Esop}
  This difference comes from the fact that the semantics condition for
  closure checking~(Definition \ref{def/vc/semantics}) includes both
  a variance check, which is monotonic in the
  context~(Lemma \ref{lem/monotonicity}) and the decomposability property,
  which is anti-monotonic~(Lemma~\ref{lem/anti-monotonicity}), so the
  present judgment must be invariant with respect to the context.
\end{version}

The most interesting rule is \Rule{sc-Constr}. It checks first that
the head type constructor is $v$-closed (according to Definition~\ref
{def/v-closed}); then, it checks each subtype for decomposability from
$v$ to $v'$ \emph{with compatible witnesses}, that is, in an
environment family $\G_i$ that can be zipped into a unique environment
$\G$.

In order to connect the syntactic and semantics versions of
decomposability, we define the interpretation $\sem{\G \der \tau :
  v \To v'}$ of syntactic decomposability.

\begin{definition}[Interpretation of syntactic decomposability]
\label{def/cc/semantics}\noindent\\
We write $\sem{\G \der \tau : v \To v'}$ for the conjunction of
properties $\sem{\G \der \tau : v}$ and $\G \Der \tau : v \vto v'$.
\end{definition}

Note that our interpretation $\Gamma \der \tau : v \To v'$ does not
coincide with our previous decomposability formula $\Gamma \Der \tau :
v \leadsto v'$, because of the additional variance-checking hypothesis
that makes it composable. The distinction between those two notions of
decomposition is not useful to have a sound criterion, but is crucial
to be complete with respect to the criterion of Simonet and Pottier,
which imposes no variance checking condition.
\begin{lemma}[Soundness of syntactic decomposability]
\label {lem/decomposability-soundness}\noindent\\
If the judgment $\G \der \tau : v \To v'$ holds, then $\sem{\G \der
\tau : v \To v'}$ is holds.
\end{lemma}
\begin{proof}{}
\locallabelreset
The proof is by induction on the derivation $\G \der \tau : v \To
v'$~\llabel H.  Expanding $\sem{\G \der \tau : v \To v'}$, we must show both
$\sem{\G \der \tau : v}$, or equivalently $\G \der \tau : v$~\llabel {C1}
and $\G \Der \tau : v \vto v'$~\llabel {C2}, which itself expands to:
$$
    \forall (\br : \G), \sigma',\quad
      \app \sigma \br \prec_{v} \sigma' \implies
      \exists (\br' : \G),\uad 
        \br \prec_\G \br'
        \wedge \app \sigma {\br'} \prec_{v'} \sigma'
$$
Let  $\br$, $\tau'$ be such that $\app{\tau}{\br} \prec_v \tau'$. 
We must exhibit a sequence $\br'$ such
that $\br \prec_\G \br'$~\llabel {Cr} and 
$\app{\tau}{\br'} \prec_{v'} \tau'$~\llabel {Ct}.
Cases where the derivation of \lref H ends with \Rule{sc-Triv} and
\textsc{sc-Var} cases are direct: take $\br$ and $(\dots, \sigma', \dots)$
for $\br'$, respectively.

In the remaining cases, the derivation ends with Rule \Rule{sc-Constr} and
$\tau$ is of the form $\ty{\bs}{t}$.  
\begin{itemize}
\item 
The $v$-closure assumption of the left
premise ensures that $\tau'$ is itself of the form $\ty{\bs'}{t}$ for some
sequence of closed types $\bs'$.
By inversion on the variance $\bar{w\alpha}$ of the head constructor
$\tyc{t}$, we deduce $\app{\sigma_i}{\br} \prec_{\varcomp v {w_i}}
\sigma'_i$ for all $i$~\llabel A. 
\item
The middle premise is the zipping assumption on the contexts $\G =
\mathop{\zip}_\iI \G_i$~\llabel C. 
\item
The right premises gives us subderivations $\G_i \der
\sigma_i : \varcomp v {w_i} \To \varcomp {v'} {w_i}$. 
This implies $\G_i \der \sigma_i : \varcomp v {w_i}$, for all $i$, 
which implies $\G \der \ty{\bs}{t} : v$, \ie. \lref {C1}.
By induction hypothesis, this also implies 
$\G_i \der \sigma_i : \varcomp v {w_i}$~\llabel D and 
$\G_i \Der \sigma_i : \varcomp v {w_i} \To \varcomp {v'} {w_i}$ for
all $i$~\llabel E. 
\end{itemize}
We may now apply zip soundness (Lemma~\ref{lem/zip-soundness}) with
hypotheses \lref C, \lref D and \lref E, which  gives us the simultaneous
decomposition 
$\G \Der (\sigma'_i : \varcomp v {w_i} \vto \varcomp {v'}
{w_i})_\iI$.   Expanding this property (Definition \ref
{def/simultaneous-decomposition}), we may apply to \lref A to get 
to get a witness $\br'$ such that both $\br' \prec_\G \br$, \ie.
our first goal \lref {Cr},
and $(\forall \iI, \app{\sigma_i}{\br'} \prec_{\varcomp {v'} {w_i}}
\sigma'_i)$, which implies 
$\app{(\ty{\bs}{t})}{\br'} \prec_{v'} \ty{\bs'}{t}$, \ie. our second goal
\lref {Ct}.
\qed
\end{proof}

Completeness is the general case is however much more difficult and we
only prove it when the right-hand side variance $v'$ is $\rel\veq$. In
other words, we take back the generality that we have introduced
in~\S\ref{sec/expressing_decomposability} when defining
decomposability.
\begin{version}{\Not\Esop}
The proof requires several auxiliary lemmas; it is the
subject of the next subsection.
\end{version}

\begin{version}{\Esop}
\begin{lemma}[Completeness of syntactic decomposability]
\label{lem/decomposability-completeness}
If $\sem{\G \der \tau : v \To v'}$ holds for $v' \in \{\veq,\virr\}$,
then $\G \der \tau : v \To v'$ is provable.
\end{lemma}

Lemma~\ref{lem/decomposability-completeness} is an essential 
piece to finally turn the correctness criterion \Rule{req-SP} of Simonet and
Pottier into a purely syntactic criterion.
\begin{theorem}[Algorithmic criterion]
The  \Rule {req-SP} criterion is equivalent to
$$
\exists \G, (\G_i)_i,\quad
  \G \der \tau : \rel\vplus
  \ \ \wedge\ \ 
  \G = \mathop{\zip}_i \G_i
  \ \ \wedge\ \ 
  \forall i,\,\G_i \der T_i : v_i \To \rel\veq
\nobelowdisplayskip
$$
\end{theorem}
\end{version}

\begin{version}{\Not\Esop}

\subsection{Completeness of syntactic decomposability}

\label {sec/zip-completeness}

We first show a few auxiliary results that will serve in the proof of
zip completeness, and later, to reconnect our closure-checking
criterion~(Definition \ref{def/cc/semantics}) with the
full criterion of Simonet and Pottier~(\Rule{req-SP}).

\begin{lemma}[Intermediate value]\label{intermediate_value_lemma}
  Let  $\app{\tau}{\ba}$ be a type expression and 
  $\br_1$, $\br_2$, $\br_3$ three type families such that
  $\app{\tau}{\br_1} \leq \app{\tau}{\br_2} \leq \app{\tau}{\br_3}$ and 
  $\br_1 \prec_\G \br_3$ holds for some $\G$.
Then, there   exists a type family $\br_2'$ such that both
$\br_1 \prec_\G \br_2' \prec_\G \br_3$ and
$\app{\tau}{\br_2'} = \app{\tau}{\br_2}$ hold.
\end{lemma}

\begin{proof}{}
\locallabelreset
We reuse the notations of the definition and assume
$\app{\tau}{\br_1} \leq \app{\tau}{\br_2} \leq \app{\tau}{\br_3}$~\llabel {T123} and 
$\br_1 \prec_\G \br_3$~\llabel {G13}. 
We just have to exhibit $\rho_2'$ such that both
$\br_1 \prec_\G \br_2' \prec_\G \br_3$\llabel {C123} and
$\app{\tau}{\br_2'} = \app{\tau}{\br_2}$~\llabel {C22'} hold.
Let $\Delta$ be the most general variance of $\tau$, \ie.  the
lowest context such that $\Delta \der \tau : \vplus$~\llabel {Delta} holds. 
By principal inversion~(Corollary \ref{cor/principal-inversion})
applied to \lref {T123} twice thanks to \lref {Delta}, we have
$\br_1 \prec_\Delta \br_2 \prec_\Delta \br_3$~\llabel {br123}.

If $\Delta \geq \G$, the result is immediate, as 
$\br_1 \prec_\G \br_2 \prec_\G \br_3$ follows from
\lref {br123} by anti-monotonicity and both \lref {C123} and \lref {C22'}
hold when we take $\br_2$ for $\br_2'$.
Otherwise, we  reason on each variable of the context $\G$ independently. 
\XXX{If we moved this \wlg. assumption up front, we would have simpler notations}
We may assume, \wlg., that $\tau$ is defined over
a single free variable $\alpha$, and $\G$ and $\Delta$ are
single variances $v_\Delta, v_\G$ with
$v_\Delta \ngeq v_\G$. We reason by case analysis on the
possible variances for $(v_\Delta, v_\G)$, which are
$\{(\any,\veq), (\vplus,\vminus), (\vminus,\vplus), (\virr,\any)\}$.

If $v_\G$ is $\rel\veq$, the hypotheses \lref {G13} 
and  \lref {T123} become $\rho_1 = \rho_3$ and
$\app{\tau}{\rho_1} \leq \app{\tau}{\rho_2} \leq \app{\tau}{\rho_1}$, 
which implies $\app{\tau}{\rho_1} = \app{\tau}{\rho_2}$. Thus, taking
$\rho_1$ for $\rho_2'$ satisfies \lref {C123} and~\lref {C22'}. 

If $v_\Delta$ is $\rel\virr$, then by irrelevant of $\virr$ and \lref
{Delta}, we have $\app{\tau}{\rho_1} = \app{\tau}{\rho_2}$. Thus, taking
$\rho_1$ for $\rho'_2$ satisfies \lref {C123} and~\lref {C22'},  as above. 

Finally, the cases $(\vplus,\vminus)$ and $(\vminus,\vplus)$ are
symmetric and we will only work out the first one, \ie. 
$v_\Delta$ is $\rel\vplus$ and $v_\G$ is $\rel\vminus$. From
$\app{\tau}{\rho_1} \leq \app{\tau}{\rho_3}$ (which follows from \lref {T123}
by transitivity) and \lref {Delta}, we have 
$\rho_1 \prec_{v_\Delta} \rho_3$, \ie. $\rho_1 \leq \rho_3$.
Since the hypothesis~\lref {G13} becomes $\rho_1 \geq \rho_3$, we have
$\rho_1 = \rho_3$~.  Then, taking $\rho_1$ for $\rho_2'$, 
\lref {C123} trivially holds while \lref {C22'} follows from \lref {T123}. 
\qed
\end{proof}

The next lemma connects the monotonicity of the variance-checking
judgment (checking variance at a lower context provides more
information, and is therefore harder) and the anti-monotonicity of the
decomposability formula (decomposing to a higher context provides more
information, and is therefore harder): for a fixed type expression,
the contexts at which you can check variance are higher than the
contexts at which you can decompose. This property, however, only
holds for non-trivial decomposability results (otherwise any context
can decompose): we must decompose from a $v$ to a $v'$ that do not
verify $v \geq v'$, and no variable of the typing context must be
irrelevant.

\begin{lemma}
\label{lem/order-decomposition-variance}
Let $\app{\tau}{\ba}$ be a type and $v$ and $v'$ be variances such that $v
\ngeq v'$.
If ${\G \Der \tau : v \vto v'}$ and $\Delta$ is the most general
context such that $\Delta \der \tau : v$, then, for each non-irrelevant
variable $\alpha$ of $\Delta$, we have $\G(\alpha) \leq \Delta(\alpha)$.
\end{lemma}

\begin{proof}{}
\locallabelreset
We show that the $\G$ is lower
than the lowest possible $\Delta$, \ie. it is the most general context
such that $\Delta \der \tau : v$ holds. 
Without loss of generality, we
consider the case where $\tau$ has a single, non-irrelevant variable
$\alpha$, and $\G$ and $\Delta$ are singleton contexts over
a single variance, respectively $w_\G$ and $w_\Delta$, with
$w_\Delta \neq \rel\Join$.

Therefore, we assume ${(w_\G\alpha) \Der \tau : v \vto v'}$~\llabel
{H1} and 
$(w_\Delta\alpha) \der \tau : v$~\llabel {H2}. 
We  prove that $w_\G \leq w_\Delta$. 
We actually show that for any
$\rho_1, \rho_2$ such that $\rho_1 \prec_\Delta \rho_2$ we also have
$\rho_1 \prec_\G \rho_2$~\llabel C.

From $\rho_1 \prec_\G \rho_2$, we can deduce
$\app{\tau}{\rho_1} \prec_v \app{\tau}{\rho_2}$~\llabel {T12}. 
Applying \lref {H1}, 
we get a $\rho'$ such that $\rho_1 \prec_\G \rho'$~\llabel r and
$\app{\tau}{\rho'} \prec_{v'} \app{\tau}{\rho_2}$~\llabel {Tp2}. We then reason by case
analysis on $v \ngeq v'$, considering the different cases
$\{(\virr, \any), (\vplus, \vminus), (\vminus, \vplus), (\any, \veq)\}$.

If $v$ is $\rel\virr$, the most general  $w_\Delta$ is
$\virr$, a case we explicitly ruled out: there is
nothing to prove.

If $v'$ is $\rel\veq$, then \lref {Tp2} implies 
$\app{\tau}{\rho'} = \app{\tau}{\rho_2}$, and in particular
$\rho' = \rho_2$; our goal \lref C follows from \lref r.

If $(v,v')$ is $(\vplus,\vminus)$ (we won't repeat the
symmetric case $(\vminus,\vplus)$), then \lref {T12} and \lref {Tp2}
becomes
$\app{\tau}{\rho_1} \leq \app{\tau}{\rho_2}$ and 
$\app{\tau}{\rho_2}  \leq\app{\tau}{\rho'}$.
Given \lref r,  the intermediate value lemma (Lemma
\ref{intermediate_value_lemma}) 
ensures the existence of a $\rho''$ such that $\rho \prec_\G \rho''
\prec_\G \rho'$ and $\app{\tau}{\rho''} = \app{\tau}{\rho_2}$. From there, we
deduce $\rho'' = \rho_2$, our goal \lref C follows from \lref r, as in the
previous case.
\qed
\end{proof}

Finally, the following auxiliary lemmas will be useful in the proof of
completeness.

\begin{lemma}
  \label{lem/non-irrelevant-equality}
  If the principal variance $w$ such that 
  $(w\alpha) \der \app{\tau}{\alpha} : v$ holds is not the irrelevant variance
  $\virr$, then
  $\app\tau{\rho_1} = \app\tau{\rho_2}$ implies $\rho_1 = \rho_2$.
\end{lemma}
\begin{proof}{}
  Whatever $v$ is, $\app{\tau}{\rho_1} = \app{\tau}{\rho_2}$ implies
  both $\app{\tau}{\rho_1} \prec_v \app{\tau}{\rho_2}$ and its converse
  $\app{\tau}{\rho_2} \prec_v \app{\tau}{\rho_1}$. 
This holds in particular for the principal variance  $w$ such that 
$(w\alpha) \der \app{\tau}{\alpha} : v$. 
Moreover, by principal inversion (Corrolary~\ref{cor/principal-inversion})
applied twice, we have both $\rho_1 \prec_w \rho_2$ and $\rho_2 \prec_w
\rho_1$. 
If $w$ is distinct from $\virr$ this implies $\rho_1 = \rho_2$.
\qed
\end{proof}

\begin{lemma}
\label{lem/principal-equal-decomposability}
If $\Gamma \Der \tau : v \vto \rel\veq$ holds for some $\Gamma$, and
$\Delta$ is the most general context such that $\Delta \der \tau : v$
holds, then $\Delta \Der \tau : v \vto \rel\veq$ also hold.
\end{lemma}
\begin{proof}{}
\locallabelreset
Assume $\Gamma \Der \tau : v \vto \rel\veq$, \ie. 
$$
    \forall \br \bs', \app{\tau}{\br} \prec_v \bs' \implies
      \exists \br',\uad \br \prec_\G \br' \wedge \app{\tau}{\br'} = \sigma'
\eqno \llabel{Hdecomp}
$$
Assume that $\Delta$ is principal for  $\Delta \der \tau : v$~\llabel D.
We show  $\Delta \Der \tau : v \vto \rel\veq$, \ie. 
$$
    \forall \br \bs', \app{\tau}{\br} \prec_v \bs' \implies
      \exists \br',\uad \br \prec_\Delta \br' \wedge \app{\tau}{\br'} =
      \sigma'
\eqno \llabel 0
$$
Let $\br$, $\bs'$ be such that have 
$\app{\tau}{\br} \prec_v \sigma'$~\llabel{Hprec}.  
By~\lref{Hdecomp}, there exists
$\br'$ such that $\app{\tau}{\br'} = \sigma'$~\llabel =.
To prove~\lref 0, it only remains to prove that  $\br \prec_\Delta
\br'$~\llabel C.  
Given~\lref =, the inequality~\lref {Hprec} becomes $\app{\tau}{\br}
\prec_v \app{\tau}{\br'}$~\llabel 4.  Then \lref C follows 
by  principal inversion (Corrolary~\ref{cor/principal-inversion})
applied to \lref 4, given \lref D. 
\qed
\end{proof}

\begin{lemma}
\label{lem/principal-context-for-eq-or-irr}
Let $\Delta$ be the most general context such that $\Delta \der \tau : v$
holds. If $v$ is $\rel=$, then only variances $\rel\veq$ or $\rel\virr$ may
appear in $\Delta$.  If $v$ is $\virr$, then only $\rel\virr$ may appear in 
$\Delta$.
\end{lemma}
\begin{proof}{}
\locallabelreset
  If $v$ is $\rel\virr$, the context $\Gamma$ with all variances set
  to $\virr$ satifies $\Gamma \der \tau : v$ (as $\sem{\G \der
    \tau : v}$ holds). By principality we have $\Delta \leq
  \Gamma$, so $\Delta$ also has only irrelevant variances as $\rel\virr$
  is the minimal variance.

  If $v$ is $\rel\veq$, we handle each variable of the context
  independently, that is we can assume, \wlg.\XXX{Need justification}, that
  $\tau$ has only one variable $\alpha$. So $\Delta$ is of the form
  $(w\alpha)$ and we know that for any $\rho$, $\rho'$ such that $\rho
  \prec_w \rho'$ we have $\app{\tau}{\rho} = \app{\tau}{\rho'}$~\llabel H. 
If $w$ is not $\rel\virr$, by lemma~\ref{lem/non-irrelevant-equality}
applied to \lref H, we have $\rho = \rho'$ for any $\rho \prec_w \rho'$,
which means that $w$ is $\rel\veq$.
Summing up, we have shown that $w$ is either $\rel\virr$ or $\rel\veq$.
Reasoning similarly in the general case, any variance $w$ of $\Delta$ is
either $\rel\virr$ or $\rel\veq$.
\qed
\end{proof}

We can now prove the converse of the zip soundness (Lemma
\ref{lem/zip-soundness}) that is the core of the future proof of
completeness of the decomposability judgment $\G \der \tau : v \To \vprime{}$.
\begin{theorem}[Zip completeness] 
\label{thm/zip-completeness}
Given any context $\G$, a family of type expressions
$(\app{T_i}{\bs})_\iI$ and a family of variances $(v_i)_\iI$, if the simultaneous
decomposition $\G \Der (T_i : v_i \vto \vprime{i})_\iI$ holds, then there
exists a family of contexts $(\G_i)_\iI$ such that 
$\G \leq \mathop{\zip}_\iI \G_i$ and 
both $\G_i \der T_i : v_i$ and $\G_i \Der T_i : v_i \vto \vprime{i}$
hold for all $i$.

If furthermore $\G \der T_i : v_i$ holds for all $i$, then
$\mathop{\zip}_\iI \G_i$ is precisely $\G$.
\end{theorem}

\begin{proof}{}
\locallabelreset
Let us assume the simultaneous decomposition ${\G \Der (T_i : v_i
\vto \vprime{i})_\iI}$, which expands to: 
$$
\forall \bs', \br,\;
   (\forall i, \app{T_i}\br \prec_{v_i} \sigma'_i)
   \implies 
   \exists \br',\uad
    \br \prec_{\G} \br' \wedge
    (\forall i, \app{T_i}{\br'} \precprime{i} \sigma'_i)
\eqno \llabel{Hyp}
$$
We construct a family of contexts $(\G_i)_\iI$ such that the following holds:
\begin{mathpar}
\G \leq \mathop{\zip}_i \G_\iI 
~\llabel {zip}

\forall i,\,\G_i \der T_i : v_i
~\llabel {ResZip}

\forall i, \G_i \Der T_i : v_i \vto \vprime{i}
~\llabel {ResDecomp}
\end{mathpar}
where \lref {ResDecomp} is equivalent to
$$
\forall i,\uad  \forall \sigma'_i, \br,\uad
   \bigparens {
   \app{T_i}\br \prec_{v_i} \sigma'_i
   \implies
   \exists \br',\uad 
        \br \prec_{\G_i} \br'
        \wedge
        \app{T_i}{\br'} \precprime{i} \sigma'_i
   }
\eqno\llabel{ResDecomp'}
$$
The first step is to move from the entailment~\lref {Hyp} of the form
$\forall \br, (\forall \iI, \dots) \implies (\exists \br', (\forall \iI, \dots))$
to the weaker form 
$(\forall \iI, \forall \br, \dots \implies \exists \br', \dots)$, 
but closer to \lref {ResDecomp'}.
More precisely, we show that 
\[
  \forall i, \uad \forall \sigma'_i,\br,\uad
  \bigparens {
  \app{T_i}\br \prec_{v_i} \sigma'_i \implies
  \exists\br',\uad
  \br \prec_{\G} \br'  \wedge
   \app{T_i}{\br'} \precprime{i} \sigma'_i 
  }
  \eqno\llabel{EasyPart}
\]
that is, $\forall i, {\G \Der T_i : v_i \vto \vprime{i}}$~\llabel {Easy}.
Let $i$, $\sigma'_i$, and $\br$ be such that 
$\app{T_i}{\br} \prec_{v_i} \sigma'_i$~\llabel {Tisi'}.  
We show that there exists
a $\br'$ such that $\app{T_i}{\br} = \sigma'_i$~\llabel {Ti=si'} and
$\br \prec_{\G} \br'$~\llabel{r<Gr'}. 
Let us extend our type $\sigma'_i$ to
a family $\bs'$,  defined by taking
 $\sigma'_j$ equal to $\app{T_j}{\br}$ for $j$ in $I \setminus \set i$. 
By construction, we have
$(\forall \jI, \app{T_j}{\br} \prec_{v_j} \sigma'_j)$.
Therefore, we may apply~\lref{Hyp} to get a $\br'$ such that
$\br \prec_{\G} \br'$, \ie. our first goal \lref {r<Gr'}, and 
$(\forall j, \app{T_j}{\br'} = \sigma'_i)$, which implies our second goal
\lref {Ti=si'} when $j$ is $i$. This proves \lref{EasyPart}. 

We now prove that we can refine this to have $\br \prec_{\G_i} \br'$
for $(\G_i)_\iI$ such that $\G \leq \mathop{\zip}_i \G_i$. 

Let $\D_1$ and $\D_2$ be the most general contexts such both that $\D_1 \der T_1
: v_1$ and $\D_2 \der T_2 : v_2$ hold~\llabel D.  
Let \lref {zip}', \lref {ResZip}', and \lref {ResDecomp}'  be
obtained by replacing $\Gamma_i$'s by $\Delta_i$'s in our three goals \lref
{zip}, \lref {ResZip}, and \lref {ResDecomp}.  In fact \lref {ResZip}' is just
\lref D.
By Lemma~\ref{lem/principal-equal-decomposability} applied
to~\lref {Easy} twice, given~\lref D, we have both
$\D_1 \Der T_1 : v_1 \vto \vprime{1}$ and
$\D_2 \Der T_2 : v_2 \vto \vprime{2}$, that is, \lref {ResDecomp}.

Hence $\D_1$ and $\D_2$ are correct choices for $\G_1$ and $\G_2$
if they also satisfy the goal~\lref{zip}', \ie. $\D_1 \zip \D_2
\geq \G$. 
We now study when remaining goal \lref {zip}' holds and, when it does not,
propose a different choice for $\G_1$ and $\G_2$ that respect all three
goals.

\Wlg., we assume that $I$ is reduced to $\set {1, 2}$ and that 
there is only one free variable $\beta$ in $T_1, T_2$.
Since we focus on a single variable of the context we name $w_1$ and $w_2$
the variances of $\beta$ in $\D_1$ and $\D_2$, \XXX{Changed $\G_i$ to
$\D_i$} respectively. We now reason by case analysis on the variances $w_1$
and $w_2$.

If both of them are $\virr$, we have $\D_1 \zip \D_2 = (\virr\beta)$,
so we do not necessarily have $\G \leq \D_1 \zip \D_2$. 
Instead, we make a different choice for $G_2$. Namely, we  pick $\G$ for
$\G_2$ and keep $\D_1$ for  $\G_1$.  
As $\D_2$ is $\virr$ we have $\D_2 \leq \G_2$, so by
monotonicity of the variance checking judgment we have $\G_2 \der
T_2 : v_2$ from~\lref D, and we still have $\G_2 \Der T_2 : v_2 \vto
\rel\veq$ from~\lref{Easy}. Hence~\lref {ResZip} and \lref {ResDecomp}
are reestablished.
Finally, we have $\G_1 \zip \G_2 = \G$, so in
particular $\G \leq \G_1 \zip \G_2$, \ie. \lref {zip}. 

If only one of the $w_i$ is $\virr$, we may assume, \wlg., that it is
$w_1$.  Then $\D_1 \zip \D_2$ is $\D_2$ and we only need to show that 
$\G \leq \D_2$~\llabel {GD2}. 
From \lref {Easy}, we have $\G \Der T_2 : v_2 \To
\vprime{2}$. We then make a case analysis on $v_2$:
if $v_2$ is not $\rel\veq$, then 
by 
since $\D_2$ is most general and $w_2 \neq \virr$, we may apply
Lemma~\ref{lem/order-decomposition-variance} 
to get \lref {GD2};
Otherwise, $v_2$ is $\rel\veq$;
Lemma~\ref{lem/principal-context-for-eq-or-irr} applied to \lref D implies
that $\D_2$ is itself $(\veq\beta)$, and \lref {GD2} trivially
holds. \XXX{Because $\rel\veq$ is the top variance. Recheck as this was not
clearly the argument used.}

Finally, if none of the $w_i$ is $\virr$, we first prove that they are
both $\rel\veq$. In fact, we only prove that $w_1$ is $\rel\veq$~\llabel
{w1}, as the other case follows  by symmetry. 
To prove~\lref {w1}, we assume that $\rho''$ be such that $\rho
\prec_{w_1} \rho''$  and we show that $\rho = \rho''$~\llabel {LC} holds.  
By~\lref D, we have $\app{T_1}{\rho} \prec_{v_1} \app{T_1}{\rho''}$. 
By reflexivity, we have $\app{T_2}{\rho} \prec_{v_2} \app{T_2}{\rho}$.
We can use those
two inequalities to invoke our simultaneous decomposability
hypothesis~$\lref{Hyp}$ with $\app{T_1}{\rho''}$ for $\sigma'_1$ 
and $\app{T_2}{\rho}$ for $\sigma'_2$ 
to get a $\rho'$ such that both
$\app{T_1}{\rho'} \precprime{1} \app{T_1}{\rho''}$ and 
$\app{T_2}{\rho'} \precprime{2} \app{T_2}{\rho}$ hold.
By Lemma~\ref {lem/non-irrelevant-equality}\XXX[WAS]{\ref
{lem/principal-equal-decomposability}}
applied with \lref D, this implies both $\rho' = \rho''$ and $\rho' = \rho$, 
and therefore~\lref {LC}. 

Therefore, the only remaining case is when $w_1$ and $ w_2$ are both
$\rel\veq$. Then $\D_1 \zip \D_2$ is $(\veq\beta)$, which is the highest
single-variable context. So our goal \lref {zip}' trivially holds. 

Of these several cases, one ($\virr,\virr$) has $\G_1 \zip \G_2 =
\G$ directly, and in the others $\G_1$ and $\G_2$ were defined as the
most general contexts such that $\G_1 \der T_1 : v_1$ and $\G_2 \der
T_2 : v_2$. If we add the further hypothesis that for each $i$, $\G
\der T_i : v_i$ holds, then by principality of the $\G_i$, we have that
$\G_i \leq \G$ for each $i$. This implies that we have
$(\mathop{\zip}_\iI \G_i) \leq \G$ (when it is defined, $\zip$
coincides with the lowest upper bound $\wedge$).
By combination with \lref {zip}, we  get  $(\mathop{\zip}_\iI \G_i) = \G$.
\qed
\end{proof}

\begin{lemma}[Completeness of syntactic decomposability]
\label{lem/decomposability-completeness}\noindent\\
If $\sem{\G \der \tau : v \To v'}$ holds for $v' \in \{\veq,\virr\}$,
then $\G \der \tau : v \To v'$ is provable.
\end{lemma}
\begin{proof}{}
\locallabelreset
Assume $\sem{\G \der \tau : v \To v'}$ holds for $v' \in \{\veq,\virr\}$,
\ie. 
$\G \der \tau : v$~\llabel d
and
$\G \Der \tau : v \vto v'$~\llabel D, which 
expands to
$$
\let \sigma \tau
\forall (\br : \G), \sigma',\uad
      \app \sigma \br \prec_{v} \sigma' \implies
      \exists (\br' : \G),\uad
        \br \prec_\G \br'
        \wide\wedge \app \sigma {\br'} \prec_{v'} \sigma'
\eqno \llabel {D'}
$$
We show  $\G \der \tau : v \To v'$~\llabel C  by 
structural induction on $\tau$

If $v \geq v'$ holds, then~\lref C
directly follows from Rule \Rule{sc-Triv}.
This applies in particular when $v' = \virr$. 
Hence, we only need to consider the remaining cases where
$v'$ is $\rel\veq$ and $v \not \geq v'$


We now reason by cases on $\tau$. 

\Case {Case $\tau$ is a variable $\alpha$.} \lref {D'} becomes
$$
\forall \rho, \tau',\uad
      \rho \prec_{v} \tau' \implies
      \exists \rho',\uad
        \rho \prec_\G \rho'
        \wide\wedge \rho' = \tau'
$$
This means that if $\rho \prec_v \tau'$ holds then $\rho \prec_{\G} \tau'$ also
holds: the variance $w\alpha \in \G$ satisfies $v \geq w$. 
Since, the hypothesis \lref d implies $v \leq w$, we have
$v = w$. Therefore, \lref C follows by Rule \Rule{sc-Var}.

\Case {Case $\tau$ is of the form $\ty{\bs}{t}$.}
By inversion, the derivation of~\lref d must end with rule \Rule
{vc-Constr}, hence we have  $\G \der \sigma_i : \varcomp v
{w_i}$~\llabel{HypVar} for each $i \in I$ with  $\G \der
\mc{type}~\ty{\bar{w\alpha}}{t}$~\llabel t.   
 

  \newcommand{\Iirr}{{I_{\virr}}}
  \newcommand{\Inirr}{I_{\not \virr}}

Let us show that $\G \Der (\sigma_i : \varcomp v {w_i} \vto {\varcomp \veq
{w_i}})_\iI$~\llabel S, \ie. 
$$
    \forall \br, \bs',\uad
      \parens 
        {\forall \iI, \app{\sigma_i}{\br} \prec_{v_i} \sigma'_i}
      \implies
      \exists \br',\uad \br \prec_\G \br'
        \wedge \parens 
         {\forall \iI, \app {\sigma_i} {\br'} \prec_{v'_i} \sigma'_i}
$$
Let $\br$ and $\bs'$ be such that  
$\app{\sigma_i}{\br} \prec_{v_i} \sigma'_i$ holds for all $i$ in $I$. 
From this and \lref {HypVar}, we have $\app {(\ty\bs t)} \br \prec_{v}
\ty{\bs'} t$. By applitcation of \lref D, there exists
$\br'$ such that $\br \prec_\G \br'$ and
$\app {(\ty\bs t)} {\br'} \prec_{\veq} \ty{\bs'} t$. 
By inversion of subtyping\XXX{Attention, not by the inversion lemma...}
\XXX{This is just mentioned in plain text, but it be a lemma},
this implies $\app{\sigma_i}{\br'} \prec_{\varcomp \veq {w_i}} \sigma'_i$, 
for all $i$ in $I$. This proves~\lref S.
We also note that the constructor $\tyc{t}$ is $v$-closed~\llabel {Note}.

To prove our goal~\lref C, we construct a family $(\G_i)_\iI$ of contexts that satisfies
$\mathop{\zip}_\iI \G_i = \G$~\llabel G and subderivations $\G_i \der
  \sigma_i : \varcomp v {w_i} \To \varcomp \veq {w_i}$~\llabel i, since then
the conclusion~\lref C follows by an application of
rule \Rule {sc-Constr} with~\lref t, \lref G, and \lref i.

We will handle separately the arguments $\sigma_j$ that are irrelevant, \ie.
when $w_j$ is $\virr$,  from the rest. 
Let $\Iirr$ be the set of indices with irrelevant variances and $\Inirr$ the
others. 
 
For any $i \in \Iirr$, $\varcomp v {w_i}$ and $\varcomp {\veq} {w_i}$ are
both $\virr$ so the condition~\lref i, which becomes $\G_i \der \sigma_i :
\virr \To \virr$, is void of content ($\app{\sigma_i}{\br}
\prec_{\virr} \sigma'_i \implies \exists \br',\uad \app{\sigma_i}{\br'}
\prec_{\virr} \sigma'_i$ is always true). More precisely, let $\Gamma_i$ be
the irrelevant context having only irrelevant variances.  Then \lref i
follows by Rule \Rule {sc-Triv}.

Since the decomposability constraints for $i \in \Iirr$ such that
$w_i = \virr$ are trivial, \lref S  is equivalent to $\G \Der (\sigma_i : \varcomp
v {w_i} \vto {\varcomp \veq {w_i}})_{i \in \Inirr}$~\llabel {Snirr}. 
\XXX{Better show it, formally.}

For each $i \in \Inirr$, $\rel{\varcomp \veq {w_i}}$ equals $\rel\veq$, so
\lref {Snirr} becomes $\G \Der (\sigma_i : \varcomp v {w_i} \vto
\rel\veq)_{i \in \Inirr}$.
By zip completeness (Theorem~\ref{thm/zip-completeness}), there is a family
$(\Gamma_i)_{i \in \Inirr}$ such that
$\G \leq \mathop{\zip}_{i \in \Inirr} \Gamma_i$ and
both $\G_i \der \sigma_i : \varcomp v {w_i}$ and $\G_i \Der \sigma_i :
\varcomp v {w_i} \vto \rel\veq$, \ie. 
$\sem{\Gamma_i \der \sigma_i :\varcomp v {w_i} \To \rel\veq}$~\llabel h hold for any
$i \in \Inirr$, 
Furthermore, since we also
have~$\lref{HypVar}$, we can strengthen our result into $\mathop{\zip}_{i \in \Inirr} 
\Gamma_i = \G$.  By induction hypothesis applied to \lref h, we have 
\lref i for $i \in \Inirr$. 

We have two families of contexts over domains $\Iirr$ and $\Inirr$ that
partition $I$; we can union them in a family $(\Gamma_i)_{\iI}$ that has
subderivations $\forall \iI, \G_i \der \sigma_i : \varcomp v {w_i} \To
\varcomp {v'} {w_i}$. As the contexts in $\Iirr$ are all irrelevant, they
are neutral for the zipping operation: $\mathop{\zip}_{i \in I} \Gamma_i$
is equal to $\mathop{\zip}_{i \in \Inirr} \Gamma_i$, that is $\G$. 
This proves \lref G while \lref i has already been proved separately for 
$i \in \Iirr$ and $i \in \Inirr$.  
\qed
\end{proof}

\begin{remark}[Note \lref {Note}]
   The head
    constructor $\tyc{t}$ is closed in our system with atomic
    subtyping, but the situation is in fact a bit stronger than that:
    the statement of $v$-closure of $\ty{\ba}{t}$ can be formulated in
    term of decomposability $\Gamma \Der \ty{\ba}{t} : v \vto
    \rel\veq$. It is very close from our decomposability hypothesis
    $\Gamma \Der \ty{\bs}{t} : v \vto \rel\veq$, but uses variables
    $\ba$ instead of full type expressions $\bs$. We conjecture that
    the decomposability hypothesis (with $\veq$ on the right) implies
    $v$-closure in a much larger set of subtyping systems that just
    atomic subtyping: it suffices that the subtyping relation is
    defined only in term of head constructors.
\end{remark}

\XXX{It would be useful so say that all cases where the zip is defined 
have been covered and that the only missing part for full completeness
is to show that cases where the zip is underfined are not possible.}

\subsection{Back to the correctness criterion}

\locallabelreset

Remember the correctness criterion \Rule{req-SP} of Simonet and Pottier: 
\[
\forall \bs, \bs', \br, \quad
\left(
   \ty{\bs}{t} \leq \ty{\bs'}{t}
     \wedge \app D{\bs,\br}
   \implies
   \exists \br',\; \app D{\bs',\br'}
     \wedge \app\tau\br \leq \app\tau{\br'}
\right) 
\eqno \llabel 1
\]
We now show how the closure judgment $\G \der \tau : v \To v'$ can be
used to verify that this criterion holds: we will express this
criterion in an equivalent form that uses the interpretation of our
judgments.

The first step is to rewrite the property $\ty{\bs}{t} \leq
\ty{\bs'}{t}$ using the variance annotation $\Gva$ of
$\tyc{t}$. Again, we are taking the variance annotation for the
datatype $\tyc{t}$ as granted (this is why we can use it in this
reasoning step), and checking that the definitions of the constructors
of $\tyc{t}$ are sound with respect to this annotation. \[
\forall \bs, \bs', \br, \quad
\left(
   (\forall i, \sigma_i \prec_{v_i} \sigma'_i)
      \wedge \app D{\bs,\br}
   \implies
   \exists \br',\; \app D{\bs',\br'}
     \wedge \app\tau\br \leq \app\tau{\br'}
\right) 
\eqno \llabel 2
\]
Since, the constraint $\app D{\ba,\bb}$ is a set of equalities of the form
$\alpha_i = \app{T_i}\bb$ (where $T_i$ is a type), \lref 2 is actually:
\[
\forall \bs, \bs', \br, \uad
   \parens {\forall i, \sigma_i \prec_{v_i} \sigma'_i}
     \wedge \parens {\forall i, \bs_i = \app{T_i}{\br}}
   \implies
   \exists \br',\; (\forall i, \app{T_i}{\br'} = \sigma'_i)
     \wedge \app\tau\br \leq \app\tau{\br'}
\]
Substituting the equalities and, in particular, removing the quantification
on the $\bs$, which are fully determined by the equality constraints $\bs =
\app{\bar{T}}\br$, we get:
\[
\forall \bs', \br, \quad
   (\forall i, \app{T_i}\br \prec_{v_i} \sigma'_i)
   \implies
   \exists \br',\; (\forall i, \app{T_i}{\br'} = \sigma'_i)
     \wedge \app\tau\br \leq \app\tau{\br'}
\eqno\llabel 3
\]
By inversion (Theorem~\ref{thm/inversion}), we may replace the goal $\app \tau
\br \leq \app \tau {\br'}$ by the formula
$\exists \G, (\G \der \tau : \vplus) \wedge (\br \prec_\G \br')$.
Moreover, since $\G \der \tau : \vplus$ always for for some $\G$, we may 
move  this quantification in front. Hence, \lref 3  is equivalent to:
\[
\exists \G, \uad\bigwedge 
\begin{cases}
\G \der \tau : \vplus
\\
\forall \bs', \br,\;
   (\forall i, \app{T_i}\br \prec_{v_i} \sigma'_i)
   \implies
   \exists \br',\; (\forall i, \app{T_i}{\br'} = \sigma'_i)
     \wedge \br \prec_{\G} \br'
\end{cases}
\eqno\llabel 4
\]
We may recognize in second clause the simultaneous 
decomposability judgment (Definition~\ref {def/decompasability})
$\G \der (T_i : v_i \vto \veq)_{\iI}$. Hence, \lref 4 is in fact: 
$$
\exists \G, \uad
\G \der \tau : \vplus
\wide\wedge \G\der (T_i : v_i \vto \veq)_{\iI}
\eqno\llabel {4b}
$$
Then, comes the delicate step of this series of equivalent rewriting:
\[
\exists \G, (\G_i)_\iI, \uad
\bigwedge
\begin{cases}
\G \der \tau : \vplus
\\
     \G = \mathop{\zip}_\iI \G_i
     \wide\wedge
      \forall i, \uad
      \bigparens {\G_i \der T_i : v_i
      \wide\wedge
      \G_i \der T_i : v_i \vto \veq}
\end{cases}
\eqno\llabel 5
\]
The reverse imiplication from \lref 5 to \lref {4b} is is the zip soundness
(Lemma \ref{lem/zip-soundness}).  

The direct implication, from \lref {4b} to
\lref 5 is more involved: let $\G_0$ be such that $\G_0 \der \tau
: \vplus$.  By zip completeness (Theorem
\ref{thm/zip-completeness}), with the hypotheses of $\lref{4}$, there
exists a family $(\G_i)_\iI$ satisfying the typing, zipping and
decomposability of second line of \lref{5} with $\G_0 \leq \mathop{\zip}_\iI
\G_i$. We take $\mathop{\zip}_\iI \G_i$ for $\G$. Then, from $\G_0 \leq \G$
we get $\G \der
\tau : \vplus$ by monotonicity (Lemma~\ref {lem/monotonicity}).

As a last step, the last conjuncts of \lref 5 are equivalent to $\forall
i,\,\G_i \der T_i : v_i \To \rel\veq$ by interpretation of syntactic
decomposability (Definition~\ref {def/cc/semantics}) and soundness and
completeness of zipping (lemmas~\ref {lem/decomposability-soundness}
and~\ref {lem/decomposability-completeness}).  Therefore, \lref 5 is
equivalent to:
\[
\exists \G, (\G_i)_\iI,\quad
  \G \der \tau : \rel\vplus
  \ \ \wedge\ \ 
  \G = \mathop{\zip}_\iI \G_i
  \ \ \wedge\ \ 
  \forall \iI,\,\G_i \der T_i : v_i \To \rel\veq
  \qquad
\eqno \llabel 7
\]
which is our final criterion. 

\end{version}

\Esop{}{\paragraph{Pragmatic evaluation of this criterion}}

This presentation of the correctness criterion only relies on
syntactic judgments. It is pragmatic in the sense that it suggests
a simple and direct implementation, as a generalization of the check
currently implemented in type system engines --- which are only
concerned with the $\Gamma \der \tau : \vplus$ part.

To compute the contexts $\Gamma$ and $(\Gamma_i)_\iI$ existentially
quantified in this formula, one can use a variant of our syntactic
judgments where the environment $\Gamma$ is not an input, but an
output of the judgment; in fact, one should return for each variable
$\alpha$ the \emph{set} of possible variances for this judgment to
hold. For example, the query $(? \der \alpha * \ty{\beta}{ref}: \vplus)$
should return $(\alpha \mapsto \{\vplus, \veq\}; \beta \mapsto
\{\veq\})$. Defining those algorithmic variants of the judgments is
routine, and we have not done it here\Esop{by lack of space}{}.
The sets of variances corresponding to the decomposability of the
$(T_i)_\iI$ ($? \der T_i : v_i \To \rel\veq$) should be zipped
together and intersected with the possibles variances for $\tau$,
returned by
($? \der \tau : \vplus$). The algorithmic criterion is satisfied if
and only if the intersection is not empty; this can be decided in
a simple and efficient way.

\section{Closed-world vs. open-world subtyping}

\begin{version}{\Not\Esop}
\subsection{Upward and downward closure in a ML type system}

In the type system we have used so far, \emph{all types} are both
upward and downward-closed. Indeed, thanks to the simplicity of our
subtyping relation, we have a very strong inversion principle: two
ground types in a subtyping relation necessarily have exactly the same
structure. We have therefore completely determined a sound variance
check for a simple type system with \GADT.

This simple resolution, however, does not hold in general: richer subtyping
relations will have weaker invertibility properties. As soon as a bottom
type $\bot$ is introduced, for example, such that that for all
type $\sigma$ we have $\bot \leq \sigma$, downward-closure fails
for most types. For example, products are no longer downward-closed: $\G \der
\sigma * \tau \geq \bot$ does not imply that $\bot$ is equal to some
$\sigma' * \tau'$. Conversely, if one adds a top type $\top$, bigger than
all other types, then most type are not upward-closed anymore.

In OCaml, there is no $\bot$ or $\top$ type\footnote{A bottom type
  would be admissible, but a top type would be unsound in OCaml, as
  different types may have different runtime
  representations. Existential types, that may mix values of different
  types, are constructed explicitly through a boxing step.}. 
However, object
types and polymorphic variant have subtyping,  so they are, in general,
neither upward nor downward-closed. Finally, subtyping is also used
in private type definitions, that were demonstrated in the example.

Our closure-checking relation therefore degenerates into the
following, quite unsatisfying, picture:
\begin{itemize}
\item no type is downward-closed because of the existence of private types;
\item no object type  but the empty object type is upward-closed;
\item no arrow type is upward-closed because its left-hand-side would
  need to be downward-closed;
\item datatypes are upward-closed if their components types are.
\end{itemize}
From a pragmatic point of view, the situation is not so bad; as our
main practical motivation for finer variance checks is the relaxed
value restriction, we care about upward-closure (covariance) more than
downward-closure (contravariance). This criterion tells us that
covariant parameters can be instantiated with covariant datatypes
defined from sum and product types (but no arrow), which would satisfy
a reasonably large set of use cases.
\end{version}

\subsection{A better control on upward and downward-closure}

As explained in the introduction, the problem with the upward and
downward closure properties is that they are not monotonic: enriching
the subtyping lattice of our type system does not preserve them. While
the core language has a nice variance check for \GADT, adding private
types in particular destroys the downward-closure property of the
whole type system.

Our proposed solution to this tension is to give the user the choice
to locally strengthen negative knowledge about the subtyping relation
by abandoning some flexibility. Just as object-oriented languages have
a concept of \code{final} classes that cannot be extended, we would
like to allow to define \code{downward-closed} datatypes, whose private
counterparts cannot be declared, and \code{upward-closed} datatypes
that cannot be made \code{invisible}: defining
$\mc{type}~\tyc{t}~\mc{=}~\mc{private}~\tau$ would be rejected by the
type-checker if $\tau$ was itself declared \code{downward-closed}.

\begin{version}{\Not\Esop}
Such ``closure specifications'' are part of the semantic properties of
a type and would, as such, sometimes need to be exposed through module
boundaries. It is important that the specification language for
abstract types allow to say that a type is upward-closed
(respectively downward-closed). These new ways to classify types raise
some software engineering questions. When is it desirable to define
types as \code{upward-closed}? The user must balance its ability to
define semi-abstract version of the type against its use in a \GADT---and
potentially other type-system features that would make use of
negative reasoning on the subtyping relation. We do not yet know how
to answer this question and believe that more practice is necessary
to get a clearer picture of the trade-off involved.
\end{version}

\subsection{Subtyping constraints and variance assignment}
\label{sec/gadts-with-subtyping-constraints}

\begin{version}{\Not\Esop}
We will now revisit our previous example, using the guarded existential notation:
\begin{lstlisting}
  type $\ty{\alpha}{expr}$ =
    | Val of $\exists \beta [\alpha = \beta].\,\beta$
    | Int of $[\alpha = \tyc{int}].\,\tyc{int}$
    | Thunk of $\exists \beta \gamma [\alpha = \gamma].\, \ty{\beta}{expr} * (\beta \to \gamma)$
    | Prod of $\exists \beta\gamma[\alpha = \beta*\gamma].\, \ty{\beta}{expr} * \ty{\gamma}{expr}$
\end{lstlisting}
\end{version}
\begin{version}{\Esop}
  Consider our introductory example $\ty{\alpha}{expr}$ of strongly
  typed expressions~(\S\ref{expr-example}).
\end{version}
A simple way to get such a type to be covariant would be, instead of
proving delicate, non-monotonic upward-closure properties on the tuple
type involved in the equation $\alpha = \beta * \gamma$, to
\emph{change} this definition so that the resulting type is obviously
covariant:
\begin{lstlisting}
  type $\ty{\vplus\alpha}{expr}$ =
    | Val of $\exists \beta [\alpha \geq \beta].\,\beta$
    | Int of $[\alpha \geq \tyc{int}].\,\tyc{int}$
    | Thunk of $\exists \beta \gamma [\alpha \geq \gamma].\, \ty{\beta}{expr} * (\beta \to \gamma)$
    | Prod of $\exists \beta\gamma[\alpha \geq \beta*\gamma].\, \ty{\beta}{expr} * \ty{\gamma}{expr}$
\end{lstlisting}
We have turned each equality constraint $\alpha = \app T \bb$ into
a subtyping constraint $\alpha \geq \app T \bb$. For a type $\alpha'$
such that $\alpha \leq \alpha'$, we get by transitivity that
$\alpha' \geq \app T \bb$. This means that $\ty{\alpha}{expr}$
trivially satisfies the correctness criterion from Simonet and
Pottier. Formally, instead of checking
$\G \der T_i : v_i \To \rel\veq$, we are now checking
$\G \der T_i : v_i \To \rel\vplus$, which is significantly easier to
satisfy\Esop{}{\footnote{
    Note that the formal proofs of the precedent section were, in some
    cases, specialized to the equality constraint. More precisely, our
    decomposability criterion is still sound when extended to
    arbitrary subtyping constraints, but its completeness is unknown
    and left to future work.
}}: when $v_i$ is itself $\vplus$ we can directly apply the
\Rule{sc-Triv} rule.

While we now have a different datatype, which  gives us a weaker
subtyping assumption when pattern-matching, we are still able
to write the classic function
$\mc{eval} : \ty{\alpha}{expr} \to {\alpha}$, because the constraints
$\alpha \geq \tau$ are in the right direction to get an $\alpha$ as
a result.
\begin{lstlisting}
  let rec eval : $\ty{\alpha}{expr} \to \alpha$ = function
    | Val $\beta$ (v : $\beta$) -> (v :> $\alpha$)
    | Int (n : int) -> (n :> $\alpha$)
    | Thunk $\beta \gamma$ ((v : $\ty{\beta}{expr}$), (f : $\beta \to \gamma$)) ->
      (f (eval v) :> $\alpha$)
    | Prod $\beta$ $\gamma$ ((b : $\ty{\beta}{expr}$), (c : $\ty{\gamma}{expr}$)) ->
      ((eval b, eval c) :> $\alpha$)
\end{lstlisting}

\begin{version}{\Not\Esop}
  We conjecture that moving from an equality constraint on the \GADT
  type parameters to a subtyping constraint (bigger than, or smaller
  than, according to the desired variance for the parameter) is often
  unproblematic in practice. In the examples we have studied, such
  a change did not stop functions from type-checking---we only needed
  to add some explicit coercions.
\end{version}

However, allowing subtyping constraints in \GADTs has some
disadvantages. If the language requires subtyping casts to be
explicit, this would make pattern matching of \GADT syntactically
heavier than with current \GADTs where equalities constraints are used
implicitly.
\begin{version}{\Not\Esop}
  This is related to practical implementation questions,
  as languages based on inference by unification tend to favor
  equality over subtyping, bidirectional coercions over unidirectional
  ones.
\end{version}
Subtyping constraints need also be explicit in the type declaration,
forcing the user out of the convenient ``generalized codomain type''
syntax.

From a theoretical standpoint, we think there is value in exploring
both directions: experimenting with \GADTs using subtyping constraints,
and with fine-grained closure properties for equality
constraints. Both designs allow to reason in an open world setting, by
being resilient to extensions of the subtyping relation. Whether it is
possible to expose those features to the expert language user
(\eg. library designers) without forcing all users to pay the
complexity burden remains to be seen.


\section{Future Work}

\paragraph{Extension of the formal exposition to non-atomic subtyping}

As remarked in \S\ref{atomic_subtyping_restriction} during the
definition of our formal subtyping relation, the soundness proof of
Simonet and Pottier is restricted to atomic subtyping. We conjecture
that their work can be extended to non-atomic subtyping, and
furthermore that our results would extend seamlessly in this setting,
thanks to our explicit use of the $v$-closure hypothesis.

\begin{version}{\Not\Esop}
\paragraph{On the relaxed value restriction}

Regarding the relaxed value restriction, which is our initial
practical motivation to investigate variance in presence of \GADTs,
there is also future work to be done to verify that it is indeed
compatible with this refined notion of variance. While the syntactic
proof of soundness of the relaxation doesn't involve subtyping
directly, the ``informal justification'' for value restriction uses
the admissibility of a global bottom type $\bot$ to generalize
a covariant unification variable; in presence of downward-closed type,
there is no such general $\bot$ type (only one for
non-downward-closed types). We conjecture that the relaxed value restriction
is still sound in this case, because the covariance criterion is really used to
rule out mutable state rather than subtype from a $\bot$ type; but it
will be necessary to study the relaxation justification in more
details to formally establish this result.
\end{version}

\paragraph{Experiments with $v$-closure of type constructors as a new semantic property}

In a language with non-atomic subtyping such as OCaml, we need to
distinguish $v$-closed and non-$v$-closed type constructors. This is
a new semantic property that, in particular, must be reflected through
abstraction boundaries: we should be able to say about an abstract
type that it is $v$-closed, or not say anything.

How inconvenient in practice is the need to expose those properties to
have good variance for \GADTs? Will the users be able to determine
whether they want to enforce $v$-closure for a particular type they
are defining?

\paragraph{Experiments with subtyping constraints in \GADTs}

In \S\ref{sec/gadts-with-subtyping-constraints}, we have presented
a different way to define \GADTs with weaker constraints
(simple subtyping instead of equality) and stronger variance
properties. It is interesting to note that, for the few \GADTs
examples we have considered, using subtyping constraints rather than
equality constraints was sufficient for the desired applications of
the \GADT.

However, there are cases were the strong equality relying on
fine-grained closure properties is required. We need to consider more
examples of both cases to evaluate the expressiveness trade-off in,
for example, deciding to add only one of these solutions to an
existing type system.

\begin{version}{\Not\Esop}
\XXX[GS]{Are the implementation question below appropriate for the
  article?} On the implementation side, we suspect that adding
subtyping constraints to a type system that already supports \GADT and
private types should not require large engineering efforts
(in particular, it does not implies supporting the most general forms
of bounded polymorphism). Matching on a \GADT $\ty{\alpha}{t}$ already
introduces local type equalities of the form $\alpha = \app{T}{\bb}$ in
pattern matching clauses. Jacques Garrigue suggested that adding an
equality of the form $\alpha = \mc{private}~\app{T}{\bb}$ should
correspond to \GADT equations of the form $\alpha \leq \app{T}{\bb}$,
and lower bounds could be represented using the dual notion of
\code{invisible} types. Regardless of implementation difficulties, in
a system with only explicit subtyping coercion, such subtyping
constraints would still require more user annotations.
\end{version}

\begin{version}{\Not\Esop}
\paragraph{Mathematical structures for variance studies}

There has been work on more structured presentation of \GADTs as part
of a categorical framework (\cite{ghani-popl07} and
\cite{Hamana-Fiore-11}). This is orthogonal to the question of
variance and subtyping, but it may be interesting to re-frame the
current result in this framework.

Parametrized types with variance can also be seen as a sub-field of
order theory with very partial orders and functions with strong
monotonicity properties. Finally, we have been surprised to find that
geometric intuitions were often useful to direct our formal
developments. It is possible that existing work in these fields would
allow us to streamline the proofs, which currently are rather
low-level and tedious.

\end{version}

\paragraph{Completeness of variance annotations with domain information}
\label{future_work:domain_information}

For simple algebraic datatypes, variance annotations are ``enough'' to
say anything we want to say about the variance of
datatypes. Essentially, all admissible variance relations between
datatypes can be described by considering the pairwise variance of
their parameters, separately.

\begin{version}{\Not\Esop}
This does not work anymore with \GADTs. For example, the equality type
$\ty{(\alpha,\beta)}{eq}$ cannot be accurately described by
considering variation of each of its parameters independently. We
would like to say that
$\ty{(\alpha,\beta)}{eq} \leq \ty{(\alpha',\beta')}{eq}$ holds as soon
as $\alpha = \beta$ and $\alpha' = \beta'$. With the simple notion of
variance we currently have, all we can soundly say about \code{eq} is
that it must be invariant in both its parameters---which is
considerably weaker. In particular, the well-known trick of
``factoring out'' \GADT by using the \code{eq} type in place of
equality constraint does not preserve variances: equality constraints
allow fine-grained variance considerations based on upward or
downward-closure, while the equality type instantly makes its
parameters invariant.
\end{version}

\begin{version}{\Esop}
  This does not work anymore with \GADTs. For example, with only this
  notion of variance, all we can soundly say about the equality type
  $\ty{(\alpha,\beta)}{eq}$ is that it must be invariant in both its
  parameters. In particular, the well-known trick of ``factoring out''
  \GADT by using the \code{eq} type in place of equality constraint
  does not preserve variances.
\end{version}

We think it would possible to regain some ``completeness'', and in
particular re-enable factoring by \code{eq}, by considering
\emph{domain information}, that is information on constraints that
must hold for the type to be inhabited. If we restricted the subtyping
rule with conclusion $\ty{\bs}{t} \leq \ty{\bs'}{t}$ to only cases
where $\ty{\bs}{t}$ and $\ty{\bs'}{t}$ are inhabited---with a separate
rule to conclude subtyping in the non-inhabited case---we could have
a finer variance check, as we would only need to show that the
criterion of Simonet and Pottier holds between two instances of the
inhabited domain, and not any instance. If we stated that the domain
of the type $\ty{(\alpha,\beta)}{eq}$ is restricted by the constraint
$\alpha = \beta$, we could soundly declare the variance
$\ty{(\virr\alpha, \virr\beta)}{eq}$ on this domain---which no longer 
prevents from factoring out \GADTs by equality types.

\section*{Conclusion}

Checking the variance of \GADTs is surprisingly more difficult (and
interesting) than we initially thought.  We have studied a novel
criterion of upward and downward closure of type expressions and
proposed a corresponding syntactic judgment that is easily
implementable. We presented a core formal framework to prove both its
correctness and its completeness with respect to the more general
criterion of Simonet and Pottier.

This closure criterion exposes important tensions in the design of
a subtyping relation, for which we previously knew of no convincing
example in the context of ML-derived programming languages. We have
suggested new language features to help alleviate these tensions,
whose convenience and practicality is yet to be assessed by real-world
usage.

Considering extension of \GADTs in a rich type system is useful in
practice; it is also an interesting and demanding test of one's type
system design.


\def\urltilda{\kern -.15em\lower .7ex\hbox{\~{}}\kern .04em}
\def\urldot{\kern -.10em.\kern -.10em}
\def\urlhttp{http\kern -.10em\lower -.1ex\hbox{$\colon\!$}\kern -.12em\lower 0ex\hbox{/}\kern -.18em\lower 0ex\hbox{/}}

\bibliographystyle{alphaurl}
\bibliography{variance_gadts}

\begin{version}{\Not\Esop}

\pagebreak
\appendix
\tableofcontents
\end{version}

\end{document}